\documentclass[10pt]{article}

\usepackage[T1]{fontenc}
\usepackage[utf8]{inputenc}
\usepackage[letterpaper,margin=0.94in]{geometry}

\usepackage[bottom]{footmisc}

\usepackage{amsmath,amssymb,amsthm,amsfonts}
\allowdisplaybreaks[4] 
\usepackage{mathtools,mleftright}
\usepackage{mathrsfs}
\usepackage{latexsym}
\usepackage{bm}
\usepackage{bbm}
\usepackage{dsfont}
\usepackage{braket}
\usepackage{aligned-overset}
\usepackage{inconsolata}

\mleftright
\usepackage{graphicx}
\usepackage{float}
\usepackage{afterpage}
\usepackage{booktabs}
\usepackage{tabularx}
\usepackage{multirow}
\usepackage{tabularray}
\usepackage{array}
\usepackage{diagbox}
\usepackage{subcaption}
\UseTblrLibrary{varwidth}
\usepackage{aligned-overset}

\usepackage{enumitem}

\usepackage{algorithmic}
\usepackage{framed}
\usepackage{environ}

\usepackage{xspace}

\usepackage[table]{xcolor}
\usepackage[colorinlistoftodos]{todonotes}

\definecolor{TSUYUKUSA}{RGB}{46,169,223}
\definecolor{KURENAI}{RGB}{203,27,69}
\definecolor{hkugreen}{HTML}{49B880}
\definecolor{hkublue}{HTML}{009CDE}
\definecolor{hkudarkgreen}{HTML}{1C584B}
\definecolor{hkudarkred}{HTML}{FF665E}
\definecolor{navyblue}{HTML}{000080}
\definecolor{quantumviolet}{HTML}{771C56}
\definecolor{deepblue}{RGB}{0,51,153}
\definecolor{deepred}{RGB}{153,0,51}

\usepackage{tikz}
\usetikzlibrary{
  quantikz2,
  decorations.pathmorphing,
  calc,
  arrows.meta
}
\usepackage{pgfplots}
\pgfplotsset{compat=1.18}
\usepackage{ytableau}

\usepackage[
  pdfstartview=FitH,
  pdfpagemode=UseNone,
  colorlinks=true,
  citecolor=blue,
  linkcolor=magenta,
  urlcolor=navyblue,
  backref=page,
  linktoc=section
]{hyperref}
\usepackage[nameinlink,capitalize]{cleveref}

\newtheorem{theorem}{Theorem}[section]
\newtheorem{lemma}[theorem]{Lemma}

\newtheorem*{proposition*}{Proposition}
\newtheorem{corollary}[theorem]{Corollary}

\newtheorem{observation}[theorem]{Observation}

\theoremstyle{definition}
\newtheorem{fact}[theorem]{Fact}
\newtheorem{definition}[theorem]{Definition}

\newtheorem{remark}[theorem]{Remark}

\renewcommand{\Pr}{\mathop{\bf Pr\/}}
\newcommand{\E}{\mathop{\bf E\/}}

\renewcommand{\sp}{\mathrm{sp}}
\newcommand{\conv}{\mathrm{conv}}

\newcommand{\tr}{\mathrm{tr}}
\newcommand{\Tr}{\tr}
\newcommand{\poly}{\mathrm{poly}}

\newcommand{\sgn}{\mathrm{sgn}}

\newcommand{\diag}{\mathrm{diag}}

\newcommand{\bmt}{\bm{\theta}}
\newcommand{\Arg}{\mathrm{Arg}}

\newcommand{\R}{\mathbb{R}}
\newcommand{\C}{\mathbb{C}}
\newcommand{\N}{\mathbb{N}}
\newcommand{\Z}{\mathbb{Z}}

\newcommand{\T}{\mathbb{T}}
\newcommand{\Sym}{\mathfrak{S}}
\newcommand{\bfk}{\mathbf{k}}

\newcommand{\op}{\mathrm{op}}

\newcommand{\eps}{\varepsilon}

\newcommand{\spec}{\mathsf{spec}}
\newcommand{\Span}{\mathrm{span}}
\newcommand{\TP}{\mathsf{T}}

\newcommand{\Lin}{\mathsf{Lin}}

\newcommand{\I}{\mathrm{i}}

\renewcommand{\Ket}[1]{|#1\rangle\!\rangle}
\renewcommand{\Bra}[1]{\langle\!\langle#1|}
\renewcommand{\Braket}[1]{\langle\!\langle#1\rangle\!\rangle}

\newcommand{\su}{\mathsf{SU}}

\newcommand{\ie}{\text{i.e.}\xspace}

\newcommand{\Y}{\mathsf{Y}}
\newcommand{\dd}{\,\mathrm{d}}
\newcommand{\U}{\mathsf{U}}

\newcommand{\ad}{\mathrm{Ad}}

\renewcommand{\O}{\mathsf{O}}
\newcommand{\SO}{\mathsf{SO}}

\newcommand{\bfq}{\mathbf{q}}
\renewcommand{\H}{\mathcal{H}}
\newcommand{\openone}{\mathbb{I}}

\newcommand{\bigo}[1]{\mathcal{O}\left(#1\right)}

\newcommand{\ignore}[1]{}

\renewcommand{\eqref}[1]{\textup{(\ref{#1})}}
\newcommand{\eref}[1]{Eq.~\eqref{#1}}
\newcommand{\lref}[1]{Lemma~\ref{#1}}
\newcommand{\tref}[1]{Theorem~\ref{#1}}

\newcommand{\syrow}{%
  \tikz[baseline={([yshift=-.5ex]current bounding box.center)},scale=0.15]{%
    \draw (0,0) rectangle (2,1);
    \draw (1,0) -- (1,1);
  }%
}

\hypersetup{
  pdftitle={Query-optimal unitary channel tomography in diamond distance with parallel access},
  pdfauthor={Entong He, Zihao Li, Yuxiang Yang},
}

\title{Query-optimal unitary channel tomography in diamond distance with parallel access}

\author{Entong He\thanks{The University of Hong Kong. \href{mailto:ethe@cs.hku.hk}{\texttt{ethe@cs.hku.hk}}}
\and
Zihao Li\thanks{The University of Hong Kong. \href{mailto:zihaoli@hku.hk}{\texttt{zihaoli@hku.hk}}}
\and 
Yuxiang Yang
\thanks{The University of Hong Kong. \href{mailto:yuxiang@cs.hku.hk}{\texttt{yuxiang@cs.hku.hk}}}
}

\date{}
\allowdisplaybreaks

\begin{document}

\maketitle

\begin{abstract}
    In the study of quantum process tomography, it has remained open how to learn a unitary channel in diamond distance with strictly parallel queries as efficiently as with sequential queries. Sequential queries allow one to exploit adaptivity to refine the estimate and adjust the learning strategy accordingly, whereas such adjustments are impossible for parallel queries. From the perspective of higher-order quantum operations, any parallel learning strategy can be simulated by a sequential one, whereas the converse does not hold in general, suggesting that sequential learning strategies are potentially more powerful. Contrary to this intuition, we present a new protocol for learning an unknown $d$-dimensional unitary channel to within $\varepsilon$ in diamond distance using $\bigo{d^2/\eps}$ parallel queries, matching the lower bound presented in [Haah, Kothari, O'Donnell, and Tang, FOCS '23]. Our protocol thus closes the gap between parallel and sequential strategies for unitary channel tomography. As applications, it achieves optimal query complexity for boundary-regime quantum channel tomography and improves the query efficiency of the best-known protocol for tomography of fermionic linear optics.
\end{abstract}

\tableofcontents

\section{Introduction}
The dynamics of a $d$-dimensional closed quantum system is, from a mathematical perspective, described by a unitary channel $\mathcal{U}(\cdot) = U (\cdot) U^\dagger$. This description prompts a simple yet fundamental question for experimentalists: given black-box queries to this evolution, how efficiently can one reconstruct its action on every input quantum state $\eps$-accurately with high probability? In physicists' terminology, this is the \emph{tomography of reversible quantum processes}\footnote{
In the literature, this task also appears under the names \emph{unitary estimation}, \emph{unitary process tomography}, and \emph{unitary channel/transformation learning}. When discussing prior work, we follow the terminology of each original source.
}, with the distinction that we use the worst-case error, e.g., the diamond distance, rather than the average one. Such a guarantee applies uniformly to all input states, including those entangled with an external reference system. This topic has been extensively studied under different access models, access architectures, and performance measures \cite{PhysRevA.64.050302, PhysRevA.72.042338, Kahn_2007, PhysRevA.81.032324, BaldwinKalevDeutsch2014, Sedl_k_2019, Yang_2020, Christandl2021, ZhaoEtAl2024BoundedGate, one_to_one_correspondence2026, chen_Girardi2026}. Readers can also refer to \cite[Sections 1.2 and 1.3]{HKOT23} for a more comprehensive overview.

In 2023, the breakthrough work of Haah, Kothari, O'Donnell, and Tang \cite{HKOT23} established that $\Theta(d^2/\varepsilon)$ queries are necessary and sufficient to $\varepsilon$-accurately learn a $d$-dimensional unitary channel in diamond distance at a constant success probability. Their $\mathcal{O}(d^2/\varepsilon)$-query algorithm relies on a ``bootstrapping'' procedure that interleaves sequential queries to the unknown channel with known operations, adaptively updating both the hypothesis and the intermediate operations via feedback from measurements. A comparable $\widetilde{\mathcal{O}}(d^{2}/\varepsilon)$ query complexity is achievable given sequential, controlled access to both $U$ and $U^{\dagger}$ \cite{vanapeldoorn2022quantumtomographyusingstatepreparation}. More recently, Grewal and Liang \cite{grewal2026efficientlearningstructuredquantum} recovered the optimal $\mathcal{O}(d^2/\eps)$ query complexity for general qubit unitaries through a fine-grained analysis of the Pauli spectra of unitary operators, again relying on sequential bootstrapping.

The aforementioned algorithms thus exploit both the total number and the architecture of the queries. In a \emph{sequential} learning strategy, queries to the unknown channel may be interleaved with arbitrary quantum operations, allowing the learner to retain quantum information between queries and adaptively design subsequent operations based on historical measurement outcomes. A \emph{parallel} learning strategy prepares a joint input state, applies all queries in a single oracle layer, and performs a final measurement directly to retrieve information. The input may be entangled across the query systems and arbitrary ancillas, and the final measurement may act jointly on all outputs and retained ancillary systems. The restriction is that the output of one query cannot influence the input to another. From the perspective of higher-order quantum operations and quantum combs~\cite{Giulio_comb_2009}, any parallel learning strategy can be simulated by a sequential one, whereas the converse does not hold in general~\cite{PhysRevA.81.032339, Bavaresco2022, Liu_2023, 10.1145/3618260.3649621}, suggesting that sequential strategies are strictly more powerful for some quantum information processing tasks.

Nevertheless, parallel strategies have their own practical advantages. Parallelism is a fundamental driver of large-scale data processing: recent advances in classical computing rely heavily on distributing computations across processors and executing them concurrently \cite{lecun2015deep,hennessy2019golden}. A similar issue arises in quantum information processing, where the arrangement of queries can be as important as their number \cite{giovannetti2011advances,Giulio_comb_2009,Liu_2023}. Sequential protocols may require long coherent circuits or repeated measurement and feedback, whereas parallel protocols, when simultaneous access is available, can reduce oracle-query depth and experimental runtime. This is particularly useful when the unknown process cannot be repeated on demand, allowing the learner to probe all available resources at once without relying on subsequent adaptive access. This distinction is relevant to near-term quantum platforms \cite{preskill2018nisq}, including superconducting circuits, neutral Rydberg atoms, and photonic devices, where many degrees of freedom can be controlled simultaneously, but long sequential circuits and repeated feedback loops can impose time overhead and accumulate errors \cite{arute2019quantum,wu2021strong,bluvstein2024logical,zhong2020quantum}.

Returning to our original task of diamond-distance unitary channel tomography, we are thus motivated to ask whether the practical advantages of parallel strategies can be achieved without sacrificing query efficiency. More precisely, we seek to resolve the following dichotomy:

\begin{center}
    \textit{Can strictly parallel strategies attain the $\mathcal{O}(d^2/\varepsilon)$ query complexity achievable with their sequential counterparts for learning a $d$-dimensional unitary channel to accuracy $\varepsilon$ in diamond distance, or is there an (asymptotic) separation between these two query architectures?}
\end{center}

In this work, we resolve this dichotomy in favor of the former: with a properly chosen entangled probe state and a joint measurement, unitary channel tomography can be solved using optimal counts of parallel queries that match the query lower bound for this task \cite[Theorem 1.2]{HKOT23}, establishing that there is no asymptotic separation between the two architectures for this task in terms of query complexity.

\begin{theorem}[Query-optimal learning of unitary channel with parallel queries] \label{thm:query_optimal_parallel_tomo}
Let $d \geq 1$ and $0 < \eps < 2$. Given only parallel black-box access to an unknown $d$-dimensional unitary channel $\mathcal{U}$, there exists an algorithm that outputs a classical description of a unitary channel $\widehat{\mathcal{U}}$ as the estimate of $\mathcal{U}$, such that
$$
\Pr \left[ \left\|\widehat{\mathcal{U}} - \mathcal{U}  \right\|_{\diamond} \leq \eps \right] \geq \frac{2}{3},
$$
and uses at most $Q = d(L - d) \leq \frac{8 \pi d^2}{\eps} = \bigo{d^2/\eps}$ parallel queries to $U$, where $L := \lceil \frac{8 \pi d}{\eps} \rceil - 1$.
    
\end{theorem}

\tref{thm:query_optimal_parallel_tomo} shows that a single round of parallel queries suffices to gather all the information required for diamond-distance unitary channel tomography. There is no need to interleave the black-box queries with intermediate operations, nor to update those operations as the protocol proceeds. In other words, adaptivity is not necessary for query-optimal unitary channel tomography. This adds a new entry to the ``adaptivity-does-not-help''\footnote{
There is one subtle point to clarify. Strict parallelism is an even stronger restriction than non-adaptivity: a non-adaptive strategy may still query the oracle in sequence with fixed interleaved operations, thus maintaining coherence between consecutive queries. 
} literature discussed in \cite{10353129}.

Meanwhile, a query complexity upper bound in diamond distance also implies bounds under other error measures. For any $\widehat{U}, U \in \U(d)$, we denote the diamond distance error $e_{\diamond}(\widehat{U}, U) := \| \widehat{\mathcal{U}} - \mathcal{U} \|_{\diamond}$, the entanglement infidelity $e_{\mathrm{ent}}(\widehat{U}, U) := 1 - | \Tr( \widehat{U}^\dagger U ) |^2 / d^2$, the average gate infidelity $e_{\mathrm{avg}}(\widehat{U}, U) := 1 - (| \Tr( \widehat{U}^\dagger U ) |^2 + d) / (d(d+1))$, and the worst-case gate infidelity $e_{\mathrm{wc}}(\widehat{U}, U) := 1 - \min_{\ket{\psi}} | \braket{\psi | \widehat{U}^\dagger U | \psi } |^2$. We omit their arguments for brevity. As reviewed in \cite[Section 1.1]{HKOT23} and \cite[Section IV.C]{PRXQuantum.6.030202}, the following inequalities hold:
$$
e_{\mathrm{avg}} = \frac{d}{d+1} e_{\mathrm{ent}}, \quad e_{\mathrm{ent}} \leq \frac{1}{4} e_{\diamond}^2 = e_{\mathrm{wc}} \leq \frac{d}{2} e_{\mathrm{ent}}.
$$
Plugging in the performance guarantee of Theorem \ref{thm:query_optimal_parallel_tomo}, we can find that $\mathcal{O}(d^2/\sqrt{\delta})$ queries suffice to achieve $\delta$-approximate tomography under each of the error measures $e_{\mathrm{ent}}$, $e_{\mathrm{avg}}$, and $e_{\mathrm{wc}}$, each matching the corresponding information-theoretic lower bound presented in \cite[Theorem 12]{one_to_one_correspondence2026} and \cite[Theorem 1.2]{HKOT23}.

\paragraph{Note added.} We recently became aware of independent and concurrent work by Scully and Zhou \cite{NoamZhouNearOptimal2026} that achieves a query complexity of $\mathcal{O}(d^2 \log d / \eps)$ via a different parallel tomography protocol.

\subsection{Prior work}
To the best of our knowledge, this is the first construction of a query-optimal parallel-access learning strategy that solves the unitary channel tomography task using diamond distance as the error measure. Our work draws on a long line of research on unitary channel estimation and learning. In this section, we review the results most closely related to our setting while briefly outlining their results and underlying constructions.

\paragraph{Parallel-access process tomography.} A standard approach to unitary tomography is itself parallel: one estimates the images of basis states and of sufficiently many superpositions to recover the relative phases between columns. The unitary promise reduces diamond-distance unitary tomography to $\mathcal{O}(d)$ pure-state tomography. As illustrated in \cite[Section 1.3]{HKOT23}, na\"ively invoking a pure-state
tomography algorithm in this way yields a query complexity of
$\mathcal{O}(d^{3}/\eps^{2})$. A sharper analysis improves the query complexity to
$\mathcal{O}(d^{2}/\varepsilon^{2})$ using an ancilla-free parallel strategy \cite{HKOT23}. The ancilla-assisted approach, dating back to Leung \cite{leung2000robustquantumcomputation}, performs tomography of the Choi state using $\mathcal{O}(d^{2}/\delta)$ parallel queries to achieve an entanglement infidelity of $\delta$. A standard conversion to diamond distance again yields an overall complexity of $\mathcal{O}(d^{3}/\varepsilon^{2})$ \cite{HKOT23}. Mele and Bittel \cite{mele2026optimallearningquantumchannels} developed a parallel, non-adaptive unitary channel-learning method whose rank-one specialization achieves $\mathcal{O}(d^{2}/\varepsilon^{2})$ query complexity through a direct analysis of the channel's Choi estimator, which has been extended to learning quantum channels with parallel queries by invoking the random purification technique \cite{chen_Girardi2026}.

\paragraph{Parallel covariant unitary estimation under fidelity losses.} Ac\'in, Jan\'e, and Vidal \cite{PhysRevA.64.050302} gave an early finite-query representation-theoretic optimization for estimating an arbitrary $\su(d)$ transformation with respect to the Haar-averaged channel fidelity, under the assumption that copies of the unitary channel are applied in parallel on half of a bipartite system. For $\su(2)$ estimation, rotation and reference-frame estimation protocols
achieving $1/n^{2}$ error scaling with $n$ parallel queries were developed
independently by Peres and Scudo 
\cite{peres2001entangled}, Bagan
et al. \cite{BaganBaigMunozTapia2004a, BaganBaigMunozTapia2004b}, Chiribella et al. \cite{chiribella2004efficient},
and Hayashi \cite{hayashi2006parallel}, each under its own invariant cost
function. Chiribella, D'Ariano, and Sacchi established a broader covariant framework for estimating compact group transformations \cite{PhysRevA.72.042338}. Kahn \cite{Kahn_2007} subsequently proved that the $1/n^2$ error rate is optimal and gave a saturating protocol, but the optimal prefactor scaling in $d$ was left undetermined.

Bisio, Chiribella, D'Ariano, Facchini, and Perinotti \cite{PhysRevA.81.032324} showed that a parallel strategy suffices for optimal learning of unitary channels under the storage-and-retrieval criterion, measured by average fidelity, via a covariant measure-and-operate protocol. Building upon this framework, Yang, Renner, and Chiribella \cite{Yang_2020} constructed a unitary tomography algorithm with sine-weighted representation sectors that achieves expected entanglement infidelity $\mathcal{O}(d^4 / n^2)$ when $n \geq \Omega(d^2)$. However, converting this average-measure guarantee to diamond distance $\eps$ gives a $\mathcal{O}(d^{2.5}/\eps)$ query complexity \cite{HKOT23}. More recently, Yoshida, Koizumi, Studzi\'nski, Quintino, and Murao \cite{one_to_one_correspondence2026} established a correspondence between optimal parallel unitary estimation and deterministic port-based teleportation, and pinned down the tight $\Theta(d^4)$ dimension dependence of the $n^{-2}$ infidelity coefficient. Specifically, for $\su(3)$ estimation, Yoshida, Yoshida and Murao \cite{yoshida2025asymptoticallyoptimalunitaryestimation} obtained the sharper infidelity expansion via a graph-Laplacian approach. These results have progressively refined the (in)fidelity analysis of parallel strategies; upgrading the guarantee to diamond distance, however, generally requires finer-grained control of the reconstruction error.

\paragraph{Architecture optimality and separation in quantum information processing.} Parallel access concerns the causal arrangement of queries to quantum resources, not a restriction to independent experiments\footnote{
For instance, the architectural advantages quantified by quantum Fisher information in quantum metrology are asymptotic and can potentially emerge in the limit of infinitely many independent repetitions; see, e.g., \cite{PhysRevLett.123.110501}.}: a parallel unitary tomography protocol may prepare a jointly entangled probe, apply all available copies of the unknown unitary in parallel in a single oracular layer, and perform a collective measurement. Sequential protocols additionally allow quantum information to pass from an earlier oracle output to a later oracle input through intermediate operations and memory \cite{Giulio_comb_2009} and are thus more powerful. This is a distinction broader than classical feedback in choosing subsequent experiments.

These two architectures do not necessarily separate. For noiseless single-parameter unitary families, Giovannetti, Lloyd, and Maccone \cite{giovannetti2006quantum} show that entangled parallel probes can attain the Heisenberg scaling achieved by coherent sequential interrogation. Duan, Feng, and Ying \cite{PhysRevLett.98.100503} show that the perfect discrimination between two known unitary channels requires the same number of queries in both architectures.

Such equal query efficiency is, however, not universal. Bavaresco, Murao, and Quintino \cite{Bavaresco2022} exhibit a four-unitary ensemble perfectly distinguishable with two sequential queries, but not with any two-query parallel protocol. For general channels, Harrow, Hassidim, Leung, and Watrous \cite{PhysRevA.81.032339} give a pair distinguishable with two sequential calls but with no finite number of parallel ones. Strict finite-query hierarchies also appear in tasks including quantum channel discrimination \cite{PhysRevLett.127.200504}, noisy single-parameter metrology \cite{Liu_2023}, and deciding quantum promise problems \cite{10.1145/3618260.3649621}.

For learning unitary channels, group symmetry has been shown to restore the optimality of parallel strategies. Chiribella, D'Ariano, and Perinotti \cite{PhysRevLett.101.180501} established this parallel optimality for estimating a unitary representation of a compact group under group-invariant error measures, and also proved it for binary unitary discrimination. Recently, Hayashi~\cite{Hayashi2025indefinitecausal} revisited this estimation problem by introducing generalized POVMs, and showed that the conclusion holds even when indefinite-causal-order strategies are allowed.

\subsection{Applications}

\subsubsection{Query-optimal boundary-regime quantum channel tomography}
\label{subsec:boundary_regime_channel_tomography}
Our first application is an improvement to the known upper bound for quantum channel tomography in the boundary regime. For a quantum channel
$\mathcal{E}:\Lin(\mathbb C^{d_1})\to \Lin(\mathbb C^{d_2})$
with Kraus rank at most $r$, its dilation rate is defined as $\tau:= r d_2/d_1$ \cite{chen_Girardi2026}.
Since quantum channels are trace-preserving, the normalization constraint gives $\tau \geq 1$. 
The boundary regime corresponds to $\tau=1$, or equivalently $r d_2=d_1$. Chen, Girardi, Oufkir, Yu, and Zhang \cite {chen_Girardi2026} established the query complexity upper bound
$$
    \mathcal{O}\left(
        \min\left\{
        \frac{r d_1^{3/2}d_2}{\eps},
        \frac{r d_1d_2}{\eps^2}
        \right\}
    \right),
$$
for quantum channel tomography in diamond distance in this regime. They also proved a query-complexity lower bound $\Omega\left( r d_1d_2\eps^{-1}\right)$ for this task.
Thus, before our observation, there remained a gap of a factor $\sqrt{d_1}$ in the Heisenberg-scaling term for diamond-norm tomography in the boundary regime.

The key tool we use is the following ``dilation does not help'' theorem for parallel testers \cite{chen_Girardi2026,girardi2025random}. 
We state it in a form tailored to our application.

\begin{lemma}[Dilation does not help for parallel testers~{\cite[Theorem 1.5]{chen_Girardi2026}}]
\label{lem:dilation_does_not_help}
Let $\mathcal E$ be an unknown quantum channel, and $\mathsf{Dilation}_r(\mathcal E)$ be the set of its Stinespring dilations with ancilla dimension $r$. 
Suppose a channel-estimation task can be solved by a parallel tester that makes $K$ queries to an arbitrary dilation $V\in \mathsf{Dilation}_r(\mathcal E)$. 
Then the same task can be solved by a parallel tester that makes $K$ queries directly to $\mathcal E$.
\end{lemma}

We emphasize that the parallel nature of Lemma~\ref{lem:dilation_does_not_help} is essential here. 
The lemma shows that access to a Stinespring dilation does not help for \emph{parallel} testers, but it does not imply an analogous statement for sequential testers. 
Therefore, although the Haah-Kothari-O'Donnell-Tang protocol~\cite{HKOT23} is query-optimal, it is sequential and adaptive and hence cannot be directly combined with the dilation argument in Lemma~\ref{lem:dilation_does_not_help}. This is precisely where our parallel unitary tomography protocol becomes crucial: it is compatible with the dilation-does-not-help lemma. It can therefore be transferred from unitary tomography of the dilation to tomography of the original quantum channel.

\begin{theorem}[Optimal boundary-regime channel tomography]
\label{thm:boundary_regime_channel_tomography}
Let $0 < \eps < 2$, $\mathcal E:\Lin(\mathbb C^{d_1})\to \Lin(\mathbb C^{d_2})$ be an unknown quantum channel with Kraus rank at most $r$, and suppose that $r d_2=d_1$. Then there exists a parallel channel tomography protocol that outputs a classical description of a channel $\widehat{\mathcal E}$ satisfying
$$
\Pr\left[
\big\|\widehat{\mathcal E}-\mathcal E\big\|_\diamond\le \eps
\right] \geq \frac{2}{3},
$$
using
$\mathcal O\left(r d_1d_2\eps^{-1}\right)$
queries to $\mathcal E$. 
\end{theorem}

\begin{proof}[Proof of \tref{thm:boundary_regime_channel_tomography}]
Let $V\in \mathsf{Dilation}_r(\mathcal E)$ be an arbitrary Stinespring dilation of $\mathcal E$ with $r$-dimensional ancillae. 
Thus we can write $\mathcal E(\rho)=\Tr_{\mathrm{anc}} \left[V\rho V^\dagger\right]$,
where
$V:\mathbb C^{d_1}\to \mathbb C^{d_2}\otimes \mathbb C^r$
is an isometry. 
In the boundary regime $r d_2=d_1$, the input and output dimensions of $V$ are equal. 
So the dilation $V$ can be viewed directly as a $d_1$-dimensional unitary.

Suppose for the moment that we have query access to the dilated unitary channel $\mathcal V(\cdot)=V(\cdot)V^\dagger$.
Applying our parallel unitary tomography protocol in Theorem \ref{thm:query_optimal_parallel_tomo} to $V$, we can output a unitary $\widehat V$ satisfying $\|\widehat{\mathcal V}-\mathcal V\|_\diamond\le \eps$ with probability $\geq 2/ 3$ using $\mathcal O\left(d_1^2\eps^{-1}\right)=\mathcal O\left(r d_1d_2\eps^{-1}\right)$ queries to $V$. 
Define the reconstructed channel
$\widehat{\mathcal E}(\rho):=
\Tr_{\rm anc} [\widehat V\rho \widehat V^\dagger]$.
By the contractivity of the diamond norm, we have  
$$
\left\|\widehat{\mathcal E}-\mathcal E\right\|_\diamond
\le
\left\|\widehat{\mathcal V}-\mathcal V\right\|_\diamond \leq \eps.
$$
Finally, this dilation-query protocol is parallel, since our unitary tomography protocol is already parallel. By Lemma~\ref{lem:dilation_does_not_help}, it can be directly simulated by a parallel tester that makes the same number of queries directly to the unknown channel $\mathcal E$. 
This finishes the proof.
\end{proof}

\tref{thm:boundary_regime_channel_tomography} removes the $\sqrt{d_1}$ gap in the query complexity and matches the lower bound $\Omega(rd_1d_2\eps^{-1})$ in all of $r$, $d_1$, $d_2$, and $\eps$. This finally\footnote{
An earlier version of \cite{he2026optimalclassicalshadowestimation} claimed query-optimal unitary channel tomography with parallel access, and hence boundary-regime channel tomography; the authors subsequently identified a gap in their proof and revised the claim to match the best known bound of \cite{chen_Girardi2026}. Our construction is able to recover optimality.} settles the question of query-optimal boundary-regime channel tomography, which was initially left open in \cite[Section 1.4]{chen_Girardi2026}.

\subsubsection{Query-efficient tomography of fermionic linear optics}
For clarity, we start by introducing several notations: For a $d$-dimensional Hilbert space $\H$, we write $\wedge^\bullet \H := \bigoplus_{k=0}^{d} \wedge^k \H$, correspondingly, for $V \in \U(d)$, we define $\wedge^\bullet V := \bigoplus_{k=0}^d \wedge^k V$ that acts on the antisymmetric subspace as $\wedge^k V \ket{v_1 \wedge \cdots \wedge v_k} := \ket{V v_1 \wedge \cdots V v_k}$.

The second application is an improved upper bound for tomography of the family of fermionic Gaussian unitaries, commonly referred to as \emph{fermionic linear optics} (FLO) \cite{Bravyi_Lagrangian_2005, PRXQuantum.3.020328, christensen2026learningfermioniclinearoptics, braccia2026commutantfermionicgaussianunitaries}. The FLO concerns a special family of unitary operators on the $n$-particle exterior Fock space $\mathcal{F}_n := \bigoplus_{k=0}^n \wedge^k \C^n$. Let $\{e_j\}_{j=1}^n$ be an orthonormal basis for $\C^n$. The creation operator $c_j^\dagger$ acts on $\wedge^\bullet \C^n$ by $c_j^\dagger \ket{\psi} := \ket{e_j \wedge \psi}$, which can be extended linearly to the actions on $\mathcal{F}_n$ with the convention $\wedge^{n+1} \C^n = \{0\}$. The annihilation operator $c_j$ is defined as the adjoint of $c_j^\dagger$. For $j_1 < \cdots < j_k$ we have $c_{j_1}^\dagger \cdots c_{j_k}^\dagger \ket{0^n} = e_{j_1} \wedge \cdots \wedge e_{j_k}$, so the creation operators generate every particle-number sector from the vacuum. The fermionic parity operator is defined as $\mathbf{Par} := (-1)^{\sum_j c_j^\dagger c_j}$. These operators obey the canonical anticommutation relations \cite{Dereziński_Gérard_2013}
$$
\forall\, j, k \in [n], \quad \{c_j, c_k\} = \left\{ c_j^\dagger, c_k^\dagger \right\} = 0, \quad \{c_j, c_k^\dagger\} = \delta_{j, k} \openone_{\mathcal{F}_n}.
$$
Equivalently, we conduct our analysis in the Majorana basis, which is a new collection of self-adjoint operators $\{\gamma_{j}\}_{j=1}^{2n}$ from $\{c_j\}_{j=1}^n$ by
$$
   \forall \, j \in [n], \quad \gamma_{2j-1} = c_j + c_j^\dagger, ~
    \gamma_{2j} = -\I\left( c_j - c_j^\dagger \right),
$$
so that they satisfy simpler canonical anticommutation relations
$\{\gamma_j, \gamma_k\} = 2\delta_{j,k} \openone_{\mathcal{F}_n}$ for all
$j, k \in [2n]$.

Following \cite{christensen2026learningfermioniclinearoptics}, an $n$-mode active FLO is a unitary on $\mathcal{F}_n$ whose adjoint map acts linearly on the Majorana operators: For every $R \in \O(2n)$, we choose the implementation $\Phi(R)$ of $R$ on $\mathcal{F}_n$ satisfying
\begin{equation}
\label{eqn:Bogoliubov_transformations}
\Phi(R)^\dagger \gamma_j \Phi(R) = \sum_{k =1 }^{2n} R_{j, k} \gamma_k.
\end{equation}
Such a unitary exists and is unique up to an overall phase \cite{Dereziński_Gérard_2013}. Then the channels $\{ \ad_{\Phi(R)}: R \in \O(2n) \}$ are called $n$-mode \emph{active} FLO channels. The two connected components of $\O(2n)$, namely $\SO(2n)$ and $\O^{-}(2n) := \{ R \in \O(2n) : \det(R) = -1 \}$, correspond to the parity-preserving and parity-flipping actions, respectively. The induced action $\Phi(R)$ commutes with $\mathbf{Par}$ in the former case and anti-commutes with it in the latter. \emph{Passive} FLO channels have their defining unitaries being the particle-number-conserving copy of $\U(n)$ embedded in $\SO(2n)$ via the standard realification \cite[Equation (1.1)]{christensen2026learningfermioniclinearoptics}, with the fixed change of basis required by our interleaved Majorana ordering. Thus, for each $V \in \U(n)$, the corresponding passive FLO channel is $\operatorname{Ad}_{\wedge^\bullet V}$.

Several prior studies have investigated FLO tomography as a model for optical quantum processes. Oszmaniec, Dangniam, Morales, and Zimbor\'as \cite{PRXQuantum.3.020328} proposed an algorithm using $\widetilde{\mathcal{O}}(n^5/\eps^2)$ queries to learn active FLOs. Cudby and Strelchuk \cite{Cudby2024} consider learning general FLOs $\eps_{\mathrm{Fro}}$-accurately in the Frobenius distance and obtained an $\mathcal{O}(n^{13}/\eps_{\mathrm{Fro}}^4)$-query algorithm. Most recently, the unified algorithm of Christensen and Zhao \cite{christensen2026learningfermioniclinearoptics} uses $\mathcal{O}(n^3/\eps)$ queries for passive FLOs and $\widetilde{\mathcal{O}}(n^4/\eps)$ for active FLOs using only parity-respecting state preparations, controllable gates, and measurements, and at most one ancillary mode. When allowing $n$ ancillary modes, the latter query complexity can be improved to $\widetilde{\mathcal{O}}(n^3/\eps)$.

We will start by establishing the query complexity of passive FLO tomography. To establish the structure of queries to passive FLOs, we will need the following duality:

\begin{lemma}[Skew Howe duality, see, e.g., \cite{Howe1987, braccia2026commutantfermionicgaussianunitaries}] \label{lemma:skew_howe_duality}
    For integers $L > d \geq 1$, if we borrow the notations from Section \ref{sec:rep_theory} and Definition \ref{def:canonical_transformation}, the following $\U(d) \times \U(L - d)$-equivariant decomposition holds:
    $$
    \wedge^\bullet \left( \C^{d} \otimes \C^{L - d} \right) \cong \bigoplus_{\lambda \in \Gamma_{L, d} } \mathcal{W}_{\lambda}^d \otimes \mathcal{W}^{L - d}_{\lambda^{\TP}},
    $$
    where $\lambda^{\TP}$ stands for the transpose of the Young diagram $\lambda$. Consequently, for every $V \in \U(d)$, we have
    \begin{equation}
    \label{eqn:decomposition_of_passive_FLO_query}
    \left( \wedge^{\bullet} V \right)^{\otimes (L - d)} \cong \wedge^{\bullet} \left( V \otimes \openone_{\C^{L - d}} \right) \cong \bigoplus_{\lambda \in \Gamma_{L, d}  } V_{\lambda} \otimes \openone_{\mathcal{W}^{L - d}_{\lambda^{\TP}}} = \bigoplus_{\mathbf{k} \in \Omega_{L, d}  } V_{\lambda(\mathbf{k})} \otimes \openone_{\mathcal{W}^{L - d}_{\lambda(\mathbf{k})^{\TP}}}.
    \end{equation}
\end{lemma}

The expression in \eref{eqn:decomposition_of_passive_FLO_query} tells us that $L - d$ parallel queries to the unknown passive FLO are sufficient to supply all sectors $\lambda \in \Gamma_{L, d}$, or equivalently, $\mathbf{k} \in \Omega_{L, d}$, used in our construction of full $\U(d)$ unitary tomography [cf. Section \ref{sec:proof_of_main_theorem}]. Unlike Eq. (13), this representation contains all box-number sectors coherently within the same $L - d$ queries, and sets the dummy state $\ket{\eta_{\lambda(\bfk)^{\TP}}} \in (\mathcal{W}_{\lambda^{\TP}}^{L - d})^{\otimes 2}$. We will use the following probe state and covariant reference vectors:
$$
\ket{\Psi_{L, d}^{\mathrm{FLO}}} = \bigoplus_{\bfk \in \Omega_{L, d} } \frac{|a_{\bfk}|}{\sqrt{d_{\lambda(\bfk)}}} \Ket{\openone_{W_{\lambda(\bfk)}^{d}}} \otimes \ket{\eta_{\lambda(\bfk)^{\TP}}}, \quad \ket{\Psi_{\widehat{U}}^{\mathrm{FLO}}} = \bigoplus_{\bfk \in \Omega_{L, d} } \sqrt{d_{\lambda(\bfk)}} \Ket{{\widehat{U}}_{W_{\lambda(\bfk)}^{d}}} \otimes \ket{\eta_{\lambda(\bfk)^{\TP}}}.
$$
We can readily arrive at the likelihood function in \eref{eqn:probability_of_successful_event_reformulation_updated} in our proof without reformulation [cf. Lemma \ref{lemma:vanishing_of_cross_term}], so the fixed-centre tail bound in Section \ref{subsec:tail_probability_bounding} therefore applies directly. This provides a direct approach to learning $\mathcal{V}$ in diamond distance through parallel queries to $\wedge^\bullet V$.

For active parity-preserving FLOs, a standard two-copy Majorana lifting \cite{Bravyi_Lagrangian_2005, PRXQuantum.3.020328, braccia2026commutantfermionicgaussianunitaries} recovers the same structure as in the passive case. Specifically, from the original Majoranas $\{\gamma_j\}_{j=1}^{2n}$ we construct a new family of operators on $\mathcal{F}_n^{\otimes 2}$ that again satisfy the canonical anticommutation relations, and under which the transformation relation \eref{eqn:Bogoliubov_transformations} takes the same form.

\begin{lemma}[Two-copy lifting, see, e.g., {\cite[Appendix F]{PRXQuantum.3.020328}}]
\label{lemma:two_copy_lifting}
    There exists a fixed unitary $J_n : \mathcal{F}_n \otimes \mathcal{F}_n \to \mathcal{F}_{2n}$ such that for every $R \in \SO(2n)$,
    $$
    J_n \left( \Phi(R) \otimes \Phi(R) \right) J_n^\dagger
    = e^{\I \phi_R} \wedge^\bullet R
    $$
    for some phase $\phi_R \in \R$. Hence two parallel queries to an active parity-preserving FLO channel are equivalent to a single passive FLO query with the same underlying operator $R$.
\end{lemma}

\begin{proof}[Proof sketch of \lref{lemma:two_copy_lifting}]
    On the two-copy space $\mathcal{F}_n \otimes \mathcal{F}_n$, take the single-space Majorana $\{\gamma_j\}_{j=1}^{2n}$, we set the lifted operators $\gamma_{j, 1} := \gamma_j \otimes \openone$ and $\gamma_{j, 2} := \mathbf{Par} \otimes \gamma_j$. Then we define $b_j = \frac{1}{2} (\gamma_{j, 1} + \I \gamma_{j, 2})$ for each $j \in [2n]$. It is easily verified that the anticommutation relation $\{\gamma_{j, s}, \gamma_{k, t}\} = 2 \delta_{(j, s), (k, t)} \openone_{\mathcal{F}_n^{\otimes 2}}$ is satisfied, and by the parity-preserving condition, a similar transformation to \eref{eqn:Bogoliubov_transformations} is satisfied by
    $$
    (\Phi(R) \otimes \Phi(R))^\dagger \gamma_{j, r} (\Phi(R) \otimes \Phi(R)) = \sum_{k =1 }^{2n} R_{j, k} \gamma_{k, r}, \quad \forall \, j \in [2n], \, r \in \{1, 2\}.
    $$
    Therefore, we can identify $b_j$ with the standard annihilation operator on $\mathcal{F}_{2n}$ in the same way that $c_j$ is the standard annihilation operator on $\mathcal{F}_n$. This identification is precisely the unitary $J_n$.
\end{proof}
With \lref{lemma:two_copy_lifting} in hand, all that remains is to handle the parity-flipping component $\O^{-}(2n)$. This case in fact requires very little work, as is specified in the following lemma: the component $\O^{-}(2n)$ is a coset of $\SO(2n)$ determined simply by any fixed element of $\O^{-}(2n)$.

\begin{lemma}[From parity-flipping to parity-preserving FLOs, {\cite[Section 1.2]{christensen2026learningfermioniclinearoptics}}]
\label{lemma:parity_coset_equivalence}
    Let the fixed coset representative $R_0 = \mathrm{diag}(1, -1, \dots, -1)$. Then
    $$
    \Phi(R_0) = \gamma_1; \quad \O^{-}(2n) = R_0 \cdot \SO(2n).
    $$
    Thus, learning a parity-flipping FLO channel reduces to learning its parity-preserving counterpart, by left-composing each oracle query with $\ad_{\Phi(R_0)}$ and composing the output estimate with $\ad_{\Phi(R_0)}^{-1}$ once more. When the parity of $R$ is unknown, a single additional parallel query with a vacuum state with parity measurement suffices to determine it, as revealed by identity $\mathbf{Par} \Phi(R) \ket{0^n} = \det(R) \Phi(R) \mathbf{Par} \ket{0^n} = \det(R) \Phi(R) \ket{0^n}$.
\end{lemma}
Finally, the following stability lemma specifies how to set the parameters when applying the protocol of Theorem \ref{thm:query_optimal_parallel_tomo} to active FLO channel tomography.

\begin{lemma}[{\cite[Proposition 2.6]{christensen2026learningfermioniclinearoptics}}]
\label{lemma:operator_stability_with_FLO_channel}
    Let $R, R' \in \SO(2n)$, the corresponding active FLO channels satisfy
    $$
    \left\| \ad_{\Phi(R)} - \ad_{\Phi(R')} \right\|_{\diamond}
    \leq 2n \left\| R - R' \right\|_{\op}.
    $$
\end{lemma}

\begin{theorem}[Efficient tomography of fermionic linear optics with parallel access]
\label{thm:paralle_tomo_of_FLO}
    Let $n \geq 1$ and $0 < \eps < 2$. Given parallel black-box access to an unknown $n$-mode active FLO unitary channel $\ad_{\Phi(R)}$ where $R \in \O(2n)$, there exists a protocol that outputs a classical description of $\widehat{R} \in \O(2n)$ that satisfies
    $$
    \Pr\left[ \left\| \ad_{\Phi(\widehat{R})} - \ad_{\Phi(R)} \right\|_{\diamond} \leq \eps  \right] \geq \frac{2}{3},
    $$
    using $\mathcal{O}(n^2/\eps)$ queries to $\ad_{\Phi(R)}$. Within the unified framework, under each of the passive, parity-preserving, and parity-flipping promises, this protocol produces an estimate from the same family.
\end{theorem}

\begin{proof}[Proof of Theorem \ref{thm:paralle_tomo_of_FLO}]
    By the realification embedding (cf.\ \cite[Equation (1.11)]{christensen2026learningfermioniclinearoptics}) and the coset equivalence of Lemma \ref{lemma:parity_coset_equivalence}, we may assume without loss of generality that the channel being learned is the active parity-preserving FLO channel determined by an unknown $R \in \SO(2n)$. Assume that we make $2(L - 2n)$ parallel queries to the unitary $\Phi(R)$ for some $L > 2n$. Combining the tensor product lifting of Lemma \ref{lemma:two_copy_lifting} and the duality of Lemma \ref{lemma:skew_howe_duality}, applied with the Hilbert space dimension $d = 2n$, we have
    $$
    \left( \Phi(R) \otimes \Phi(R) \right)^{\otimes (L - 2n)} \cong e^{\I (L - 2n) \phi_R } \left( \wedge^\bullet R \right)^{\otimes (L - 2n)} \cong e^{\I (L - 2n) \phi_R } \bigoplus_{\lambda \in \Gamma_{L, 2n} } R_{\lambda} \otimes \openone_{\mathcal{W}_{\lambda^{\TP}}^{L - 2n}}.
    $$
    Then recall the proof of Theorem \ref{thm:query_optimal_parallel_tomo} in Section \ref{sec:proof_of_main_theorem}. Suppose we invoke the unitary tomography protocol with error threshold $\eta$ to obtain an estimate $R$ of $\widehat{U}$, and consider the rotated estimate $W = \widehat{U}^\dagger R$. Then the eigenphases of $W$ have joint density $\mu(\bmt)$ with respect to $\dd \nu^{\otimes 2n}(\bmt)$. Our learning protocol construction ensures the eigenphase concentration $\Pr[\bmt(W) \in I_{\eta}^{\times 2n}] \geq 2/3$, then the produced estimate $\widehat{U}$ satisfies
    $$
    \left\| \widehat{U} - R \right\|_{\op} = \left\| W - \openone  \right\|_{\op} = \max_{j \in [2n]} \left|  e^{\I \theta_j(W)} - 1\right| = 2 \max_{j \in [2n]} \left| \sin \left(\frac{\theta_j(W)}{2} \right) \right| \leq 2 \max_{j \in [2n]} \left| \sin \left( \frac{1}{2} \arcsin \frac{\eta}{2} \right) \right| \leq \eta.
    $$

    Then we round up the estimate to $\SO(2n)$ by simply taking $\widehat{R} := \mathrm{argmin}_{X \in \SO(2n) }\|X - \widehat{U}\|_{\op}$ whose existence is guaranteed by the compactness of $\SO(2n)$. Since the ground truth $R$ is also feasible, on the same event,
    $$
    \left\| \widehat{R} - R \right\|_{\op} \leq \left\| \widehat{R} - U  \right\|_{\op} + \left\| U - R \right\|_{\op} \leq 2 \left\| U - R \right\|_{\op} \leq 2 \eta.
    $$
    Finally, using Lemma \ref{lemma:operator_stability_with_FLO_channel}, if we use $\widehat{R}$ as the output estimate, the resulting active FLO channel deviates from the ground truth in diamond distance by at most
    $$
    \left\| \ad_{\Phi(\widehat{R})} - \ad_{\Phi(R)} \right\|_{\diamond} \leq 2n \left\| \widehat{R} - R \right\|_{\op} \leq 4 n \eta.
    $$
    Therefore, to satisfy the bounded error requirement in Theorem \ref{thm:paralle_tomo_of_FLO}, it suffices to take $\eta = (4 n)^{-1} \eps$. Recall Theorem \ref{thm:query_optimal_parallel_tomo} and Lemma \ref{lemma:two_copy_lifting}, by taking $d = 2n$, the total number of queries to $\Phi(R)$ is given by $2(L - 2n)$, and the query complexity is thus given by 
    $$
    2(L - 2n) < \frac{32 \pi n}{\eta} = \frac{128 \pi n^2}{\eps} = \mathcal{O}\left( \frac{n^2}{\eps} \right).
    $$
    For passive FLOs and active parity-flipping FLOs, the performance guarantees are satisfied automatically, potentially with some extra post-processing overhead (see passive specialization of Lemma \ref{lemma:operator_stability_with_FLO_channel} in \cite[Proposition 2.6]{christensen2026learningfermioniclinearoptics}). This concludes the proof.
\end{proof}

\subsection{Technical overview}
\label{sec:tech_overview}

We adopt the canonical parallel covariant unitary transformation learning framework \cite{PhysRevA.72.042338, PhysRevA.81.032324} where the probe state and measurement designs are reduced to choosing an appropriate probability distribution on the collection of legitimate representation sectors. By covariance, it suffices to ensure that the unitary channel estimate is $\varepsilon$-close to the identity channel in the diamond distance; that is, the success event yields $\| \mathcal{W} - \mathcal{I}\|_{\diamond} \leq \eps$. Equivalently, all its eigenphases lie within an arc of half-angle $\arcsin\frac{\varepsilon}{2}$, with the center of the arc unrestricted. Directly analyzing this joint localization event does not provide a tractable approach for optimizing the probability distribution. Instead, we use a solvable variational problem motivated by Holevo's phase estimation algorithm to discover a probe whose additional structure suffices to ensure that the desired success event is satisfied.

\paragraph{A variationally determined probe state.} To ensure the concentration of the eigenphases of the unitary estimate, a starting point is a periodic penalty for their displacement. We generalize the Holevo cost function $C_{H}(\hat\theta)=4\sin^{2}(\hat\theta/2)$ introduced for $\U(1)$ (phase) estimation \cite{Holevo2011} to its additive analogue $C_{H}(\hat{\bmt})=4\sum_{j=1}^{d}\sin^{2}(\hat\theta_{j}/2)$. It is arguably a natural choice geometrically and algebraically: around the identity channel, it penalizes squared displacement, and each term in the summand $4\sin^2(\hat \theta_j/2) = 2 - e^{\I \hat{\theta}_j} - e^{-\I \hat{\theta}_j }$ couples neighboring frequency levels. For phase estimation supported on $L$ consecutive frequency levels in space $\H_L = \Span\{\ket{j}\}_{j=0}^{L-1}$, minimizing the single-phase cost corresponds to finding the ground state of a Hamiltonian $H_L$ consisting of neighbor-hopping terms \cite{hayashi2006parallel, vanDam_2007}, the solution to which is the familiar sine-shaped state family \cite{hayashi2006parallel, Yang_2020, he2026optimalclassicalshadowestimation}. A well-optimized state with respect to this additive cost clearly encourages phase concentration and thus satisfies even the most stringent small-arc requirement imposed by the diamond distance guarantee.

Our second observation is that the likelihood function in our framework involves the characters $\chi_\lambda(\cdot)$. When the relevant Young diagrams are supported within a $d\times(L-d)$ box, the Weyl character formula~\cite{Goodman2009} reformulates the likelihood as the phase statistics of an auxiliary probe state in the antisymmetric subspace $\wedge^d \H_L$. Taking independent sine-style amplitudes on each subsystem would ignore this anti-symmetry and would not fit naturally into the canonical framework. Moreover, the cost-function minimization reduces to finding the ground state of the local Hamiltonian $\sum_{j=1}^{d}H_{L}^{(j)}$ over $\wedge^{d}\mathcal{H}_{L}$, where $H_{L}^{(j)}$ acts as $H_{L}$ on the $j$-th subsystem and trivially on the rest. This ground state is obtained by taking the exterior product of the $d$ eigenvectors of $H_{L}$ with the smallest eigenvalues. Similar approaches for unitary estimation that optimize the learning strategy via spectral-theoretic techniques also appear in \cite{Christandl2021, yoshida2025asymptoticallyoptimalunitaryestimation}.

\paragraph{Plugging into the learning framework and bounding the failure probability.} The above discussion merely addresses the geometry of eigenphase minimization in the canonical bases (called the occupation bases $\ket{\mathbf{k}}_{\wedge}$ indexed by a strictly increasing sequence $\bfk \in \{0, 1, \dots, L -1\}^d$) of $\wedge^{d}\mathcal{H}_{L}$, and has yet to show how to realize this within Bisio et~al.'s parallel learning protocol \cite{PhysRevA.81.032324}. Fortunately, there is a standard conversion $\bfk\mapsto \lambda(\bfk)$ from these bases to legitimate Young diagrams with at most $L-d$ boxes in the first row~\cite{Okounkov2001}, while ensuring that the maximum number of parallel queries is $\mathcal{O}(dL)$. The antisymmetric probe state naturally has support on Young diagrams with different numbers of boxes, and we transform it into a canonical learning framework as follows: We first sample the total number of queries and then, conditioned on the sampled value, apply the canonical learning protocol with that many queries. 

Substituting our designated probe state and POVM and then using covariance, we note that the coherence between different numbers of queries would vanish owing to the invariance of our success event under a common shift of all eigenphases. Moreover, the eigenphases of the rotated estimate $W$ fall into the success event with the same probability as an auxiliary phase random variable $\bmt$ whose density $\mu(\bmt)$ can be evaluated explicitly. The final step is to bound the failure probability from above. The success probability can be lower bounded by the probability that all coordinates of $\bmt$ drawn from this auxiliary density lie within the interval $I_{\eps} := [-\arcsin(\eps/2), \arcsin(\eps/2)]$\footnote{
This to a large extent resembles the idea behind the inequality $\| \widehat{\mathcal{U}} - \mathcal{U}\|_{\diamond} \leq 2 \|\widehat{U} - U\|_{\op} = 2\|W - \openone\|_{\op}$ \cite{HKOT23}: under the auxiliary phase random variable distribution, with high probability, the output estimate $W$ would satisfy the much stricter condition that every eigenphase concentrates around $0$.
}, and Markov's inequality bounds the probability that some eigenphase falls outside $I_{\varepsilon}$ by the expectation of any nonnegative upper bound on the indicator of this event. A caveat is that if we track $C_{H}(\bmt)$ to bound the tail according to $\mu(\bmt)$, $\E[C_{H}(\bmt)]$ gives $\mathcal{O}(d^3/L^2)$ so that the tail probability reads $\E[C_{H}(\bmt)] / \eps^2 = \mathcal{O}(d^3/L^2 \eps^2)$. A constant error probability requires $L=\Theta(d^{1.5}/\varepsilon)$, so the overall query complexity is $\mathcal{O}(d^{2.5}/\varepsilon)$ and ties that of \cite{Yang_2020, he2026optimalclassicalshadowestimation}. We instead directly count the number $N=N(\bm{\theta})$ of coordinates falling outside $I_\varepsilon$, so that $N(\bmt)=0$ gives rise to a sufficient event for the auxiliary success condition. Direct evaluation shows that $\E[N(\bmt)]=\mathcal{O}(d^{3}/(L\varepsilon)^{3})$, so setting $L=\Theta(d/\varepsilon)$ suffices. This finally grants us the desired $\mathcal{O}(d^{2}/\varepsilon)$ query complexity.

\section{Preliminaries} \label{sec:prelim}
We will use some notational conventions: Throughout, boldface stands for vectors, and for any $m \in \N$, we write $[m]$ as shorthand for the set of integers $\{1, 2, \dots, m\}$, and denote $\mathfrak{S}_m$ as the permutation group acting on $[m]$. For each element $\sigma \in \mathfrak{S}_m$, we denote its sign (or equivalently, its parity) as $\sgn(\sigma)$. Besides, we will use some basic notations in quantum information \cite{Nielsen_Chuang_2010}: First, for Hilbert spaces $\mathcal{H}_A, \mathcal{H}_B$, we write $\Lin(\mathcal{H}_A, \mathcal{H}_B)$ for the set of linear operators from $\H_A$ to $\H_B$, and write directly $\Lin(\H)$ when $\H_A = \H_B = \H$. For an integer $d \in \N$, we write $\U(d) \subseteq \Lin(\C^d)$ as the $d$-dimensional unitary group, and by default, we use $\dd U$ as the unique normalized (left- and right-) invariant Haar measure over the corresponding unitary group. Note that we would use $\T = \R/2\pi \Z$ to identify the phase circle, use the representative interval $(-\pi, \pi]$ whenever the absolute values of phases occur, and write the normalized Haar measure $\dd \nu(\theta) := \frac{1}{2\pi} \dd \theta$ on it. Correspondingly, the product measure $\dd \nu^{\otimes d}(\bm{\theta}) = \prod_{j=1}^d \dd \nu(\theta_j)$ when $\bm{\theta} = (\theta_1, \dots, \theta_d) \in \T^d$. For a unit complex number $z = e^{\I \theta} \in \C$, the argument function $\Arg(e^{\I \theta}) = \theta \in \T$ extracts its angle. For a square operator $X \in \Lin(\C^d)$, we denote its eigenvalues by $\spec(X) = \{\lambda_j\}_{j=1}^d$, sorted in an order appropriate to the context. The convex hull of a discrete point set $\{x_j\}_{j=1}^m$ is the set of the convex combinations of its members: $\conv(\{x_j \}_{j=1}^m) = \{ \sum_{j=1}^m p_j x_j:p_j \geq 0, \sum_{j=1}^m p_j = 1 \}$. The notation $\mathbb{L}^2(\mathcal{X},\dd\mu)$ stands for the inner product space of
square-integrable functions on $\mathcal{X}$ equipped with
$\left< f, g \right>:=\int_{\mathcal{X}}\overline{f(x)} g(x) \dd\mu(x)$.

For a linear operator $X$, we define the adjoint action $\ad_X : \rho \mapsto X \rho X^\dagger$, and frequently write the calligraphic $\mathcal{U}$ for $\ad_U$ when $U$ is unitary. For any linear operator $X \in \Lin(\H)$, we will explicitly specify its submatrix entries by $X = (X_{i, j})_{i \in \mathfrak{I}, j \in \mathfrak{J}}$ for some index subsets $\mathfrak{I}, \mathfrak{J} \subseteq [\dim \H]$. We define its trace norm and operator norm as $\|X\|_1 := \Tr|X|$ and $\|X\|_{\op} := \sqrt{\lambda_{\max}(X^\dagger X)}$, respectively, and its vectorization as $\Ket{X} := \sum_{i,j} X_{i,j} \ket{i}\ket{j}$. For a channel $\mathcal{C}: \Lin(\H_A) \to \Lin(\H_B)$, let $\H_R \cong \H_A$ be a reference space. Its (unnormalized) Choi operator is $\Phi_{\mathcal{C}} := (\mathcal{C}_A \otimes \mathcal{I}_R)(\Ket{I} \Bra{I}_{AR})$. If $\ket{\psi} \in  \C^d$ is a pure state, we use the convention $\psi = \ket{\psi} \bra{\psi} \in \C^{d \times d}$ to denote its density matrix. For a unitary $U \in \U(d)$, we use $\bmt(U) \in \T^{d}$ to represent its eigenphases labelled in uniformly random order. As all events we would discuss are permutation-invariant, this choice does not affect any of the probability analysis below. Finally, for a pair of quantum channels $\mathcal{E}, \mathcal{F}: \Lin(\mathcal{H}_A) \to \Lin(\mathcal{H}_B)$, we define their (unnormalized) diamond norm distance as
$$
\left\| \mathcal{E} - \mathcal{F} \right\|_{\diamond} = \sup_{\H_R} \max_{\ket{\psi}_{AR} \in \H_A \otimes \H_R } \left\| \left(\mathcal{E}_A \otimes \mathcal{I}_{R} - \mathcal{F}_A \otimes \mathcal{I}_R \right)(\psi_{AR}) \right\|_1.
$$

\subsection{Linear algebra}
In this section, we review basic definitions and lemmas required for our proof.

\begin{definition}[Determinant \cite{Horn2012}] \label{def:determinant}
    For any square linear operator $X = (X_{i, j})_{i, j \in [d]} \in \Lin(\C^d)$, its determinant is given by the Leibniz formula
    $$
    \det(X) = \sum_{\sigma \in \mathfrak{S}_d } \sgn(\sigma) \prod_{j=1}^d X_{\sigma(i), i}.
    $$
    In particular, permuting the rows or columns multiplies the determinant by the sign of the permutation.
\end{definition}

\begin{lemma}[Cauchy-Binet {\cite[Section 1.5]{Bhatia1997}}, adapted] \label{lemma:cauchy_binet}
    For any linear operators $X = (X_{i, j})_{i \in [m], j \in [k]} \in \C^{m \times k}, Y = (Y_{i, j})_{i \in [k], j \in [m]} \in \C^{k \times m}$ with $m \leq k$, for any subset $S = \{k_1, k_2, \dots, k_m\} \subseteq [k]$ sorted in order, let $X[:, S]$ and $Y[S,:]$ be the square submatrices obtained by taking the columns of $X$ and rows of $Y$ corresponding to $S$, then it holds that
    $$
    \det\left( XY \right) = \sum_{S \subseteq [k]:\,|S| = m } \det\left( X[:, S] \right) \det \left( Y[S, :] \right). 
    $$
\end{lemma}

\begin{corollary}[Laplace expansion \cite{Horn2012}] \label{corollary:laplace_expansion}
    For any linear operator $X=(X_{i, j})_{i, j \in [d]} \in \Lin(\C^d)$, we define its minor $X[i, j]$ as the operator in $\Lin(\C^{d-1})$ obtained by removing row $i$ and column $j$ of $X$, then for all $j \in [d]$, it holds that
    $$
    \det(X) = \sum_{\ell=1}^{d} (-1)^{j + \ell} X_{j, \ell} \det\left( X[j, \ell] \right).
    $$
\end{corollary}

\begin{definition}[Chebyshev polynomial of the second kind, see, e.g., \cite{Subramanian_2019}] \label{def:chebyshev_poly}
    The Chebyshev polynomials of the second kind are a sequence of polynomials $\{T_n(x)\}_{n\geq -1}$, defined on $x\in[-1,1]$ via the initial values $T_{-1}(x) = 0$, $T_0(x)=1$ and $T_1(x)=2x
$, together with the recurrence relation
$$
  \forall n \in \N, \quad  T_{n+1}(x)
    =
    2xT_n(x)-T_{n-1}(x).
$$
Equivalently, under the parametrization $x=\cos\theta$ with $\theta\in[0,\pi]$, these polynomials admit the following trigonometric representation
$$
T_n(\cos\theta) =
\frac{\sin\left((n+1)\theta\right)}{\sin\theta}, \quad \forall \, 0 < \theta < \pi,
$$
extended to both endpoints by continuity, so that $T_n(1) = n + 1$ and $T_n(-1) = (-1)^{n} (n + 1)$.
\end{definition}

\subsection{Representation theory} \label{sec:rep_theory}
To formalize our unitary learning algorithm, we present a minimal review of basic representation theory needed to understand our methodology. For a more comprehensive review of representation theory and its applications in quantum information theory, we refer readers to \cite{Goodman2009} and further to the Ph.D. theses of Harrow, Wright, and
Grinko \cite{Har05, Wri16, grinko2025mixed}.

\begin{definition}[Basics of group representation \cite{Goodman2009, Wri16, grinko2025mixed}]
    Let $G$ be a group. A (complex, unitary, finite-dimensional) representation of $G$ is a tuple $(\lambda, \mathcal{V}_{\lambda})$ where $\mathcal{V}_{\lambda}$ is a (finite-dimensional complex) vector space $\mathcal{V}_{\lambda}$, and $\lambda: G \to \U(\mathcal{V}_{\lambda})$ is a group homomorphism. We denote the dimension of this representation by $d_{\lambda} = \dim \mathcal{V}_{\lambda}$. We write $\chi_{\lambda}(g) = \Tr[\lambda(g)]$ for every $g \in G$ as the character of representation $\lambda$. For two representations $(\lambda, \mathcal{V}_{\lambda}), (\mu, \mathcal{V}_{\mu})$ of $G$, an intertwining map $T$ is a map $T: \mathcal{V}_{\lambda} \to \mathcal{V}_{\mu}$ such that $T \circ \lambda(g) = \mu(g) \circ T$ for all $g \in G$. Moreover, if $T$ is invertible, \ie, $\mu(g) = T \circ \lambda(g) \circ T^{-1}$, then we say $\lambda$ and $\mu$ are isomorphic, or notationally $\lambda \cong \mu$.
\end{definition}

\begin{definition}[Irreducible representations \cite{Goodman2009, Wri16, grinko2025mixed}]
\label{def:irreps}
    Let $(\lambda, \mathcal{V}_{\lambda})$ be the representation of a group $G$, a subspace $\mathcal{W}$ is called an invariant subspace of $\mathcal{V}_{\lambda}$ if it is invariant under the action induced by $\lambda$, \ie, $\forall g \in G$, $\lambda(g) \mathcal{W} \subseteq \mathcal{W}$. Such an invariant subspace $\mathcal{W}$ is trivial if $\mathcal{W} = \{0\}$ or $\mathcal{W} = \mathcal{V}_{\lambda}$. If $\mathcal{V}_{\lambda}$ has a non-trivial invariant subspace, we say that it is reducible, and irreducible otherwise. Correspondingly, the representation is called reducible/irreducible. We denote $\widehat{G}$ as the set of equivalence classes of all irreps of $G$.
\end{definition}

\begin{definition}[Partitions and Young diagrams \cite{Goodman2009}] \label{def:partitions_and_young_diagrams}
A partition $\lambda \vdash n$ is a restriction of the highest weight $\lambda$ such that $\lambda_1 \geq \lambda_2 \geq \cdots \geq \lambda_{k} \geq 0$ and $|\lambda| = n$. The length of a non-empty partition $\lambda$ is given by $\ell(\lambda) = \max\{j : \lambda_j > 0\}$ and if $\ell(\lambda) \leq d$, we write shorthand that $\lambda \vdash_d n$. Any partition $\lambda$ is graphically (the shape of) a Young diagram on $n$ boxes arranged in $\ell(\lambda)$ rows with $\lambda_j$ boxes on the $j$-th row. For instance, $\lambda = (2, 0) \in \Y_2^2$ corresponds to the diagram $\left(\syrow\right)$. 
We write $\Y_n^d$ for the collection of Young diagrams $\lambda \vdash_d n$.
\end{definition}

\begin{definition}[Representation of the symmetric group and general linear group \cite{Goodman2009}]
    Each partition $\lambda \in \Y_n^d$ indexes an irreducible representation of $\mathfrak{S}_n$ and a polynomial irreducible representation of $\U(d)$ when the global phase does not matter. For each $\lambda \in \Y_n^d$, suppose $(\pi_{\lambda}, \mathcal{S}_{\lambda}) \in \widehat{\mathfrak{S}_n}$ and $(\nu_{\lambda}, \mathcal{W}_{\lambda}^d) \in \widehat{\U(d)}$, the free modules $\mathcal{S}_{\lambda}$ and $\mathcal{W}_{\lambda}^d$ are known as the Specht module and the (Schur-)Weyl module, respectively, and we will use $d_{\lambda} = \dim \mathcal{W}_{\lambda}^d$ to denote the dimension of the latter one.
\end{definition}

 To fully exploit the symmetry of i.i.d. quantum resources, quantum learning algorithms often employ representation theory. Fortunately, within the scope of our construction, we will keep the role of representation theory to a minimum and use only the following duality theorem to set up our quantum learning algorithm.

\begin{lemma}[Schur-Weyl duality \cite{Goodman2009, grinko2025mixed}]
\label{lemma:schur_weyl}
    For any $n, d \in \mathbb{N}$, each permutation $\sigma \in \mathfrak{S}_n$ introduces a unitary $P_{\sigma}$ on $(\C^d)^{\otimes n}$, satisfying $P_{\sigma} P_{\tau} = P_{\sigma \tau}$ and $P_{\sigma}^\dagger = P_{\sigma^{-1}}$, and acting as
    $$
    \forall \ket{v_1, v_2, \dots, v_n} \in (\C^d)^{\otimes n}, \quad P_{\sigma} \ket{v_1, v_2, \dots, v_n} = \ket{v_{\sigma^{-1}(1)}, v_{\sigma^{-1}(2)}, \dots, v_{\sigma^{-1}(n)}}.
    $$
    Meanwhile, the natural action $U \in \U(d)$ induces the action $U^{\otimes n} \ket{v_1, v_2, \dots, v_n} = U \ket{v_1} \otimes U \ket{v_2} \otimes \cdots \otimes U \ket{v_n}$. These two actions commute: $[P_{\sigma}, U^{\otimes n}] = 0$. By Schur's lemma, the vector space $(\mathbb{C}^d)^{\otimes n}$ and the $n$-fold tensor of each $U \in \U(d)$ can be decomposed according to the irreps induced by the Young diagrams:
    $$
    (\mathbb{C}^d)^{\otimes n} \cong  \bigoplus_{\lambda \in \Y_n^d } \mathcal{W}_{\lambda}^d \otimes \mathcal{S}_{\lambda}, \quad U^{\otimes n} \cong \bigoplus_{\lambda \in \Y_n^d } U_{\lambda} \otimes \openone_{\mathcal{S}_{\lambda}}.
    $$
    The canonical bases for the Weyl and Specht modules are collectively identified as the Schur-Weyl basis \cite{Har05}. 
\end{lemma}
In particular, there exists a change-of-basis unitary $\mathcal{U}_{n, d}^{\mathrm{Sch}}$ on the space $(\C^d)^{\otimes n}$ that implements the isomorphism presented in \lref{lemma:schur_weyl}, which is called the Schur transformation \cite{Har05}. Prior work has demonstrated that for a moderate choice of $n$, such a transformation can be efficiently implemented $\eps$-accurately on a quantum circuit using $\bigo{\poly(n, \log d, \log(\eps^{-1}))}$ quantum gates from a universal gateset \cite{schur_transform_bacon, Krovi2019efficienthigh, burchardt2025highdimensionalquantumschurtransforms}.

Besides, in our analysis, we will need to have a finer-grained analysis of the linear-algebraic identities of the representation of unitary operators. Specifically, as we will see in the subsequent context, it suffices to track the statistics of several spectrum-dependent identities. Thus, we will need the following lemmas.

\begin{lemma}[Weyl character formula, {\cite[Corollary 7.1.2]{Goodman2009}}] \label{lemma:weyl_char_formula}
    For the highest weight $\lambda \vdash n$ with $\ell (\lambda) \leq d$ and any unitary $W \in \U(d)$ with eigenvalues $\spec(W) = \{e^{\I\theta_j}\}_{j=1}^d$, the trace of operator $W_{\lambda}$ is given by the following fraction of two antisymmetric functions:
    $$
  \chi_{\lambda}(W)  = \frac{ \det\left[ e^{\I \theta_j(d + \lambda_k - k)}  \right]_{1 \leq j, k \leq d} }{\det \left[ e^{\I \theta_j (d - k)} \right]_{1 \leq j, k \leq d} } = \frac{ \det\left[ e^{\I \theta_j(d + \lambda_k - k)}  \right]_{1 \leq j, k \leq d} }{ \prod_{1 \leq j < k \leq d}\left(  e^{\I \theta_j} - e^{\I \theta_k}\right) },
    $$
    where the last equality is due to the Vandermonde determinant formula. At repeated eigenvalues, the quotient is well-defined and interpreted by continuity.
\end{lemma}

\begin{lemma}[Weyl integral formula {\cite[Corollary 7.3.7]{Goodman2009}}, adapted] \label{lemma:weyl_integral_formula}
    For any Haar-integrable (that is, $\int_{\U(d)} |f(g)| \dd g < \infty$) class function $f$ over $\U(d)$, it holds that
    $$
    \int_{\U(d)} f(U) \dd  U = \int_{\U(d)} f(\diag (\spec(U)) ) \dd  U = \frac{1}{d!} \int_{\mathbb T^d} f\left( \mathrm{diag}\left( e^{\I \theta_1}, \dots, e^{\I \theta_d} \right) \right) \prod_{1 \leq j < k \leq d} \left| e^{\I \theta_j} - e^{\I \theta_k} \right|^2 \dd \nu^{\otimes d} (\bm{\theta}).
    $$
\end{lemma}

\subsection{The antisymmetric subspace}
Within the general representation theory of $\U(d)$, we focus on one special irreducible representation: the fully antisymmetric one. These properties will be useful in our subsequent discussions, where we construct probe states for learning $\U(d)$ channels by analogy with optimal phase estimation (that is, $\U(1)$ dynamics learning).

\begin{definition}[antisymmetric subspace and exterior product {\cite[Section 7.1.1]{Watrous2018}}] \label{def:anti_symm_ext_prod}
    For any Hilbert space $\H \cong \C^L$ and an integer $d \in \N$, the antisymmetric subspace of $\H^{\otimes d}$ is given by
    $$
    \wedge^d \mathcal{H} := \left\{ \ket{\psi} \in \mathcal{H}^{\otimes d}:~\forall \sigma \in \mathfrak{S}_d,~P_{\sigma} \ket{\psi} = \sgn(\sigma) \ket{\psi} \right\},
    $$
    The anti-symmetrizer is the following projection:
    $$
    \Pi_{\wedge^d} := \frac{1}{d!} \sum_{\sigma \in \S_d } \sgn(\sigma) P_{\sigma}
    $$
    that satisfies $\Pi_{\wedge^d}^2 = \Pi_{\wedge^d}$, $\mathrm{im}(\Pi_{\wedge^d}) = \wedge^d \mathcal{H}$, and $P_{\sigma} \Pi_{\wedge^d} = \sgn(\sigma) \Pi_{\wedge^d}$ for any $\sigma \in \S_d$. Given a collection of vectors $\{\ket{v_j}\}_{j=1}^d \subseteq \H$, their exterior product, with the normalization convention used here, is given by
    $$
    \ket{v_1 \wedge v_2 \wedge \cdots \wedge v_d} = \sqrt{d!} \cdot \Pi_{\wedge^d} \ket{v_1, v_2, \dots, v_d} = \frac{1}{\sqrt{d!}} \sum_{\sigma \in \S_d } \sgn(\sigma) \ket{v_{\sigma(1)}, v_{\sigma(2)}, \dots, v_{\sigma(d)}}.
    $$
\end{definition}

\begin{fact}[Alternating property] \label{fact:alternating_property}
    For any $\sigma \in \S_d$, the product $\ket{v_{\sigma(1)} \wedge \cdots \wedge v_{\sigma(d)}} = \sgn(\sigma) \ket{v_1 \wedge \cdots \wedge v_d}$. If there exists distinct indices $i , j \in [d]$ such that $\ket{v_i} = \ket{v_j}$, then $\ket{v_1 \wedge \cdots \wedge v_d} = 0$.
\end{fact}

\begin{corollary}[Exterior inner product] \label{cor:exterior_inner_product}
    Take two sets of states $\{\ket{v_i}\}_{i=1}^d, \{\ket{w_j}\}_{j=1}^d \subseteq \H \cong \C^L$, then
    $$
    \braket{v_1 \wedge v_2 \wedge \cdots \wedge v_d | w_1 \wedge w_2 \wedge \cdots \wedge w_d} = \det\left[ \braket{v_i | w_j} \right]_{1 \leq i, j \leq d}.
    $$
\end{corollary}

\begin{proof}[Proof of Corollary \ref{cor:exterior_inner_product}]
    Expanding the product using Definition \ref{def:anti_symm_ext_prod}, it follows readily that
    $$
    \begin{aligned}
        \braket{v_1 \wedge v_2 \wedge \cdots \wedge v_d | w_1 \wedge w_2 \wedge \cdots \wedge w_d} &= d! \braket{v_1, \cdots, v_d | \Pi_{\wedge^d}| w_1, \dots, w_d } \\
        &= \sum_{\sigma \in \S_d } \sgn(\sigma) \braket{v_1, \cdots, v_d | w_{\sigma^{-1}(1)}, \dots, w_{\sigma^{-1}(d)} } \\
        &= \sum_{\sigma \in \S_d } \sgn(\sigma^{-1}) \prod_{j=1}^d \braket{v_j | w_{\sigma^{-1}(j)}} = \det\left[ \braket{v_i | w_j} \right]_{1 \leq i, j \leq d},
    \end{aligned}
    $$
    finishing the proof.
\end{proof}

\begin{fact}
    The inner product defined in Corollary \ref{cor:exterior_inner_product} gives $\| \ket{v_1 \wedge v_2 \wedge \cdots \wedge v_d}  \|_2^2 = \det\left[ \braket{v_i | v_j} \right]_{1 \leq i, j \leq d}$. In particular, it yields unit norm when $\{\ket{v_j}\}_{j=1}^d$ are orthonormal since $\det\left[ \braket{v_i | v_j} \right]_{1 \leq i, j \leq d} = \det(\openone) = 1$.
\end{fact}

\begin{lemma}[Occupation basis {\cite[Proposition 7.9]{Watrous2018}}, restated] \label{lemma:occupation_basis}
    For integers $L, d \in \N$, we define the set of strictly increasing sequences 
    $$
    \Omega_{L, d} := \left\{ \bfk \in \{0, 1, \dots, L - 1\}^d:~k_1 < k_2 < \cdots < k_d  \right\}.
    $$
    And every $\bfk \in \Omega_{L, d}$ corresponds to an occupation product
    $$
    \ket{\bfk}_{\wedge} := \ket{k_1 \wedge \cdots \wedge k_d} = \frac{1}{\sqrt{d!}} \sum_{\sigma \in \S_d} \sgn(\sigma) \ket{k_{\sigma(1)}, \dots, k_{\sigma(d)}} \in (\C^L)^{\otimes d}.
    $$
    Then $\{ \ket{\bfk}_{\wedge} \}_{\bfk \in \Omega_{L, d}}$ forms a set of orthonormal basis for $\wedge^d \C^L$, and we use the convention $\bra{\bfk}_\wedge  := \ket{\bfk}_{\wedge}^\dagger$.
\end{lemma}

\begin{definition}[Wavefunction via Fourier representation {\cite[Section 1.5]{stefanucci2013nonequilibrium}}] \label{def:Fourier_wavefunction}
    Define the isometry $\mathscr{F}_L: \C^L \to \mathbb{L}^2(\T, \dd \nu)$ by $\mathscr{F}_L \ket{k} = f_k$, where $f_k(\theta) = e^{\I k \theta}$ for all $k \in \{0, 1, \dots, L - 1\}$. We will use the shorthand convention $\braket{\theta | k} = f_k(\theta)$. For any occupation basis $\bfk \in \Omega_{L,d}$, the corresponding antisymmetric wavefunction in the Fourier representation takes the form of a Slater determinant obtained by applying $\mathscr{F}_L^{\otimes d}$:
    $$
      \forall \bm{\theta} \in \T^d, \quad  \psi_{\bfk}(\bm{\theta})
        := \mathscr{F}_L^{\otimes d} \ket{\bfk}_{\wedge} = \bra{\bm{\theta}}_{\wedge} \ket{\bfk}_{\wedge} = \frac{1}{\sqrt{d!}} \sum_{\sigma \in \S_d } \sgn(\sigma) \prod_{j=1}^d e^{\I k_{\sigma(j)} \theta_j} = \frac{1}{\sqrt{d!}} 
        \det\left[e^{\I k_j \theta_i}\right]_{1 \leq i,j \leq d}.
    $$
\end{definition}

\begin{fact} \label{fact:wavefuntion_orthonormality}
    The set of wavefunctions $\{ \psi_{\bfk} \}_{\bfk \in \Omega_{L, d} }$ are mutually orthonormal on the inner product space $\mathbb{L}^2(\T^d, \dd \nu^{\otimes d})$.
\end{fact}

\begin{proof}[Proof of Fact \ref{fact:wavefuntion_orthonormality}]
    The result follows from direct evaluation: Take $\bfk, \bfk' \in \Omega_{L, d}$, we have
    $$
    \begin{aligned}
        \int_{\T^d} \overline{\psi_{\bfk}(\bm{\theta})} \psi_{\bfk'}(\bm{\theta}) \dd \nu^{\otimes d}(\bm{\theta}) &= \frac{1}{d!} \int_{\T^d} \overline{\det\left[ e^{\I k_j \theta_i} \right]_{1 \leq i, j \leq d}} \det \left[ e^{\I k_r' \theta_{\ell}} \right]_{1 \leq r, \ell \leq d} \dd \nu^{\otimes d}(\bm{\theta}) \\
        &= \frac{1}{d!} \sum_{\sigma, \tau \in \S_d} \sgn(\sigma) \sgn(\tau) \prod_{j=1}^d \int_{\T} e^{\I( k_{\tau(j)} - k_{\sigma(j)}' ) \theta_j  } \dd \nu(\theta_j) \\
        &= \frac{1}{d!} \sum_{\sigma, \tau \in \S_d} \sgn(\sigma) \sgn(\tau) \prod_{j=1}^d \delta_{k_{\tau(j)}, k'_{\sigma(j)}} = \frac{1}{d!} \cdot d! \delta_{\bfk, \bfk'} = \delta_{\bfk, \bfk'},
    \end{aligned}
    $$
    finishing the proof.
\end{proof}

\begin{remark}[Weyl character formula via wavefunctions {\cite[Chapter 7]{Goodman2009}}] \label{remark:reformulating_Weyl_character}
For each partition $\lambda \vdash_d n$ subject to $\lambda_1 \leq L - d$, it (bijectively) induces a strictly increasing sequence $\bfk =  \bfk(\lambda)$, where $k_j = \lambda_{d+1 - j} + j - 1$ for all $j \in [d]$, then $\bfk(\lambda) \in \Omega_{L, d}$ [cf. \lref{lemma:occupation_basis}]. Let $\psi_{\bfk}(\bm{\theta})$ denote the wavefunction of Definition \ref{def:Fourier_wavefunction}, then for each $\bm{\theta} \in \mathbb{T}^d$,
$$
\begin{aligned}
\det\left[ e^{\I(\lambda_j + d - j)\theta_i} \right]_{1 \leq i, j \leq d} &= (-1)^{\binom{d}{2}}  \det\left[ e^{\I(\lambda_{d+1-j} + j -1)\theta_i} \right]_{1 \leq i, j \leq d} \\
&= (-1)^{\binom{d}{2}}  \det\left[ e^{\I k_j\theta_i} \right]_{1 \leq i, j \leq d}  = (-1)^{\binom{d}{2}} \sqrt{d!}\,\psi_{\bfk(\lambda)}(\bm{\theta}).
\end{aligned}
$$
Likewise, in the same notation, the denominator is the wavefunction evaluated at the index $\bfk(\bm{0}) = (0, 1, \dots, d - 1)$ of the trivial representation:
$$
\begin{aligned}
\det\left[ e^{\I(d - j)\theta_i} \right]_{1 \leq i, j \leq d} = (-1)^{\binom{d}{2}} \det\left[ e^{\I(j-1)\theta_i} \right]_{1 \leq i, j \leq d} = (-1)^{\binom{d}{2}} \sqrt{d!}\,\psi_{\bfk(\bm{0})}(\bm{\theta}).
\end{aligned}
$$
Therefore, recall the expression in \lref{lemma:weyl_char_formula}, the character can be reformulated as
$$
\chi_{\lambda}(\bm{\theta}) = \frac{\psi_{\bfk(\lambda)}(\bm{\theta})}{\psi_{\bfk(\bm{0})}(\bm{\theta})}.
$$
\end{remark}

\begin{definition}[\cite{Okounkov2001}] \label{def:canonical_transformation}
    The equality $k_j = \lambda_{d+1-j} + j - 1$ establishes a bijection $\bfk \mapsto \lambda(\bfk)$ and $\lambda \mapsto \bfk(\lambda)$ between the lattice of strictly increasing sequences $\Omega_{L, d}$ [cf. Lemma \ref{lemma:occupation_basis}] and the set of partitions with a bounded first row and an indefinite number of total boxes,
    $$
    \Gamma_{L, d} = \left\{ \lambda: \ell(\lambda) \leq d,~\lambda_1 \leq L - d \right\}.
    $$
    Moreover, for each $n \in \N$, since $|\bfk| =\sum_j k_j = \sum_{j} \lambda_{d+1-j} + \sum_j (j-1) = |\lambda| + \binom{d}{2}$, we can define the set of occupation bases that correspond to partitions on $n$ boxes, and the subspace it generates, that is,
    $$
   \forall n \in \{0, 1, \dots, Q\}, \quad \Omega_{L, d}^{(n)} = \left\{ \bfk \in \Omega_{L, d}: |\bfk| = n + \binom{d}{2} \right\}, \quad
    \mathcal{K}_n = \Span\left\{ \ket{\bfk}_{\wedge}:~\bfk \in \Omega_{L, d}^{(n)}\right\}.
    $$
\end{definition}

\begin{corollary} \label{cor:decomp_ext_space_with_box_number}
    Let $\mathcal{K}_n$ as defined in Definition \ref{def:canonical_transformation}. Set $Q = d(L-d)$. $\Omega_{L, d}$ can be decomposed into disjoint sets $\Omega_{L, d} = \bigcup_{n=0}^{Q} \Omega_{L, d}^{(n)}$, and the exterior space $\wedge^d \H$ with $\dim \H = L$ has decomposition
    $
    \wedge^d \H = \bigoplus_{n=0}^{Q} \mathcal{K}_n
    $.
\end{corollary}

\begin{proof}[Proof of Corollary \ref{cor:decomp_ext_space_with_box_number}]
    Recall \lref{lemma:occupation_basis}, as $\Span\{\ket{\bfk}_{\wedge}\}_{\bfk \in \Omega_{L, d} } = \Span\{\ket{\bfk(\lambda)}_{\wedge}\}_{\lambda \in \Gamma_{L, d} } = \wedge^d \H$, occupation bases that correspond to different Young diagram sizes must span different subspaces. It suffices to track the lower and upper bounds of $n$. Note that the strictly increasing vector $\bfk$ must satisfy $j -1 \leq k_j \leq L - d + j +1$, since it must simultaneously accommodate $j-1$ increases while leaving room for the remaining $d-j$ entries. Therefore, summing up both sides yields $\binom{d}{2} = \sum_{j=1}^d (j-1) \leq \sum_{j=1}^d k_j \leq \sum_{j=1}^d (L - d + j + 1) = d(L-d) + \binom{d}{2}$. Both inequalities are saturable when setting $k_j = j-1$ and $k_j = L - d + j+1$, respectively.
\end{proof}

\section{Constructing the learning algorithm}
\label{sec:contructing_learning_algorithm}

\subsection{Optimal average-case learning of unitary channels}
\label{sec:bisio_learning_strategy}
Bisio, Chiribella, D'Ariano, Facchini, and Perinotti consider the following problem: How to construct a learner (or in the wording of \cite{Giulio_comb_2009}, a ``comb'') $\mathfrak{C}_{N \to M}[\cdot]$ that gets as input $N$ copies of a unitary channel and produces $M$ copies of it optimally under the averaged channel fidelity, defined by
$$
F(\mathfrak{C}_{N \to M}) = \frac{1}{d^{2M}} \int_{\U(d)} \Bra{U^{\otimes M} } \Phi_{\mathfrak{C}_{N \to M}[\mathcal{U}^{\otimes N}]} \Ket{U^{\otimes M}} \dd U.
$$

They showed at the existential level that an optimal learner may be chosen from a canonical covariant family \cite{PhysRevA.81.032324}: It suffices to consider the unknown unitary being queried in parallel to prepare a memory state; a covariant measurement of that state produces an estimate $\widehat U$, and the outcome is either retained as a classical description or used to implement the channel $\widehat{\mathcal U}$.
This framework has proved effective for the estimation of $\U(1)$
\cite{Bu_ek_1999, vanDam_2007}, $\su(2)$
\cite{chiribella2004efficient}, $\su(3)$
\cite{yoshida2025asymptoticallyoptimalunitaryestimation}, and general
$\su(d)$ unitary transformations \cite{Sedl_k_2019, Yang_2020, one_to_one_correspondence2026}. Nevertheless, covariance specifies the structure of the learner but does not determine the optimal sector-wise weights. Indeed, constructing those weights for a given estimation task remains highly nontrivial \cite{Kahn_2007, Yang_2022, yoshida2025asymptoticallyoptimalunitaryestimation}, while a universal approach to designing query-efficient unitary learning strategies has yet to be fully developed.

We now detail this canonical form. Suppose $n$ copies of the unknown unitary channel $\mathcal{U}^{\otimes n}$ are available. Recall the notations in Definition \ref{def:partitions_and_young_diagrams}, an optimal unitary learning strategy with $n$ copies is characterized by a subset $\Y\subseteq\Y_n^d$ of Young diagrams and a probability distribution over $\Y$:
$$
    \bfq=(q_\lambda)_{\lambda\in\Y}\in\mathbb R_{\ge 0}^{|\Y|}.
$$
Given such a strategy $(\Y, \bfq)$, up to a basis transformation [cf. Section \ref{sec:rep_theory}], the optimal probe state for parallel unitary storage takes the form
\begin{equation}
\label{eqn:canonical_probe_state}
    \ket{\Psi_{n, \bfq}} = \bigoplus_{\lambda\in\Y} \sqrt{\frac{q_\lambda}{d_\lambda}} \Ket{\openone_{\mathcal{W}_\lambda^d}}\otimes\ket{\xi_\lambda},
\end{equation}
where $\ket{\xi_\lambda} \in \mathcal{S}_{\lambda} \otimes \mathcal{S}_{\lambda}$ is a pure state chosen arbitrarily. After applying $U^{\otimes n}$ according to the register arrangement in the Schur-Weyl basis [cf. Lemma \ref{lemma:schur_weyl}] to the probe state, the memory state $U^{\otimes n}\ket{\Psi_{n, \bfq}}$ is immediately measured using the covariant POVM $\mathscr{M}=\{\ket{\Psi_{n, \widehat{U}}} \bra{\Psi_{n, \widehat{U}}}\dd {\widehat{U}}\}_{{\widehat{U}}\in\U(d)}$ consisting of reference frame vectors indexed continuously with the following structure
$$
    \ket{\Psi_{n, \widehat{U}}}
    =
    \bigoplus_{\lambda\in\Y}
    \sqrt{d_\lambda}\Ket{\widehat{U}_\lambda}\otimes\ket{\xi_\lambda}.
$$  
Together, the probe state $\ket{\Psi_{n, \bfq}}$ and the measurement $\mathscr{M}$ give the covariant unitary learning procedure. Specifically, for single-shot learning via the $n$-copy unitary channel, the only tunable object in our strategy is the probability distribution $\bfq$ over the Young diagram support $\Y$ (or equivalently, over $\Y_n^d$).

\subsection{The tomography problem with covariant learning scheme}
\label{sec:problem_reformulation}
In this section, we specialize the covariant learning algorithm of
Section~\ref{sec:bisio_learning_strategy} to unitary channel tomography in diamond distance (even though the algorithm was not originally designed for it). The learning procedure works in a ``measure-and-operate'' manner: Once the measurement outcome yields $\widehat{U}$,
we output a classical description $\widehat{\mathcal{U}}$ of the corresponding unitary channel as our estimate \cite{PhysRevA.81.032324, Yang_2020, He2026resource}. Conditioned on the ground truth being $U$, the likelihood density of $\widehat{U}$ with respect to the Haar measure reads
\begin{equation}
\label{eqn:bisio_likelihood}
\begin{aligned}
\Pr\left[ \widehat{U} | U \right] &=  \left| \braket{ \Psi_{n, \widehat{U}} | U^{\otimes n} | \Psi_{n, \bfq} } \right|^2 = \left| \sum_{\lambda \in \Y } \sqrt{q_{\lambda}} \Braket{ \widehat{U}_{\lambda} | U_{\lambda} } \right|^2 = \left| \sum_{\lambda \in \Y} \sqrt{q_{\lambda}} \chi_{\lambda}(\widehat{U}^\dagger U) \right|^2.
\end{aligned}
\end{equation}
If we write $W = \widehat{U}^\dagger U$, since $U \mapsto U_{\lambda}$ is a homomorphism, we have $\Pr[\widehat{U} | U] = \Pr[W | \openone] =: \Pr[W]$. Note also that under this interpretation, we have $\dd W = \dd \widehat{U}$, and by the invariance of the Haar measure and the unitary invariance of the diamond distance, as in \cite{Yang_2020}, the success event can be reformulated as
$$
\Pr\left[ \left\| \widehat{\mathcal{U}} - \mathcal{U} \right\|_{\diamond} \leq \eps \mid U \right] = \Pr\left[ \left\| \mathcal{W} - \mathcal{I}  \right\|_{\diamond} \leq \eps \mid \openone \right].
$$

\subsection{Finding good probe states via variational optimality}

The design of a proper probe state based on the canonical learning framework reviewed in Section \ref{sec:bisio_learning_strategy} for $\U(d)$ dynamics is motivated by the design for learning $\U(1)$ dynamics (or the phase estimation problem), its one-dimensional counterpart. Specifically, the optimality of learning $\U(1)$ under the Holevo cost function \cite{Holevo2011} is derived via the optimal solution of the tridiagonal eigenvalue problem subject to special boundary conditions \cite[Section III.B]{vanDam_2007}. In this section, we first review how the optimal distribution is derived from a recurrence relation in $\U(1)$ learning. We then lift that derivation to $\U(d)$ within the Bisio-Chiribella-D'Ariano-Facchini-Perinotti canonical learning framework, using a variant of the Holevo cost function.

\paragraph{The variational problem of learning $\U(1)$ dynamics: a primer.} The $\U(1)$ dynamics learning is a classic example in the study of quantum estimation theory (see, e.g., \cite{Helstrom1969}). We take $L \in \N$ as the truncation of the energy level of the physical system in the Fock representation, denoted by $\mathcal{H}_L := \Span\{\ket{j} \}_{j=0}^{L-1}$, and the phase representation of a $\U(1)$ unitary reads $U_{\theta} = e^{\I \theta H} = \sum_{j=0}^{L - 1} e^{\I j \theta} \ket{j} \bra{j}$. Any normalized probe state takes the form $\ket{\psi_L} = \sum_{j=0}^{L-1} a_j \ket{j}$ subject to $\sum_{j=0}^{L-1} |a_j|^2 = 1$. As for the figure of merit, we take the Holevo cost \cite{Holevo2011}, which is quadratic and periodic in terms of the estimate deviation $\hat{\theta} - \theta$:
$$
C_{H}(\hat{\theta} - \theta) = 4 \sin^2\left(\frac{\hat{\theta} - \theta}{2}  \right) = 2 - 2 \cos\left( \hat{\theta}  -\theta \right).
$$
Suppose we are agnostic about the prior distribution of $\theta$, we take the uniform prior, and by \cite{Bu_ek_1999, hayashi2006parallel, vanDam_2007, Holevo2011} it suffices to take the covariant measurement over the torus $\T$ [cf. Section \ref{sec:prelim}], take the POVM
$$
M(\dd \hat{\theta}) = U_{\hat \theta} \Lambda U_{\hat \theta}^\dagger \dd \nu(\hat \theta): \quad \int_{\T} M(\dd \hat{\theta}) = \openone_{\mathcal{H}_L}
$$
for some matrix $\Lambda \in \C^{L \times L} $. The orthogonality reads $\int_{\T} e^{\I(t -s)\hat \theta } \dd \nu(\hat \theta) = \delta_{t, s}$, so that for $\{M(\dd \hat{ \theta}) \}_{\hat{ \theta} \in \T}$ to be a legitimate POVM, it is required that
$$
\Lambda \succeq 0; \quad \forall t \in \{0, 1, \dots, L - 1\},~\Lambda_{t, t} = 1.
$$
 By the covariance of the scheme, it suffices to take the ground truth $\theta = 0$ so that $U_{\theta} = \openone$. Plugging in the corresponding likelihood induced by this measurement, the cost reads
\begin{equation} \label{eqn:expected_holevo_cost}
\begin{aligned}
\E\left[C_{H}(\hat{\theta})\right] &= \int_{\T} C_{H}(\hat{ \theta}) \bra{ \psi_L   } M(\dd \hat{ \theta}) \ket{\psi_L} \\
&= \int_{\T}\left( 2 - e^{\I \hat{\theta}} - e^{-\I \hat{ \theta}} \right) \bra{\psi_L} U_{\hat \theta} \Lambda U_{\hat \theta}^\dagger \ket{\psi_L} \dd \nu(\hat{\theta}) \\
&= 2 - \int_{\T} \left( e^{\I \hat{\theta}} + e^{-\I \hat{\theta}} \right) \sum_{s, t=0}^{L-1} \overline{a_s} a_t \Lambda_{s, t} e^{\I(s - t) \hat{\theta}} \dd \nu(\hat \theta) = 2 -   \sum_{s, t=0}^{L-1} \overline{a_s} a_t \Lambda_{s, t} \left(\delta_{s - t , -1} + \delta_{s-t, 1} \right) \\
&= 2 - 2 \Re \sum_{t=0}^{L-2} \overline{a_t} a_{t+1} \Lambda_{t, t + 1}.
\end{aligned}
\end{equation}
The positivity of $\Lambda$ requires that each $2 \times 2$ principal minor of $\Lambda$ should be positive, i.e.,
$$
\forall t \in \{0, 1, \dots, L - 2\}, \quad
\begin{pmatrix}
  \Lambda_{t, t} & \Lambda_{t, t+1} \\ \Lambda_{t+1, t} & \Lambda_{t+1, t+1}
\end{pmatrix} = \begin{pmatrix}
1 & \Lambda_{t, t+1} \\ \Lambda_{t+1, t} & 1
\end{pmatrix} \succeq 0 \quad \iff \quad  1 - |\Lambda_{t, t+1}|^2 \geq 0.
$$
Write each amplitude $a_t = r_t e^{\I \alpha_t}$, we have the basic algebraic inequality $\Re \left(\overline{a_t} a_{t+1} \Lambda_{t, t + 1} \right) \leq |\overline{a_t} a_{t+1}||\Lambda_{t,t+1}| \leq r_t r_{t+1}$. Note that this bound is saturable by taking $\Lambda_{s, t} = e^{\I(\alpha_s - \alpha_t)}$, so that is can be readily decomposed into $\Lambda = \ket{\Psi} \bra{\Psi}$ where $\ket{\Psi} = \sum_{j=0}^{L-1} e^{\I \alpha_j } \ket{j}$. Substituting this into \eref{eqn:expected_holevo_cost}, we have
$$
\E\left[C_{H}(\hat{\theta})\right] = 2 - 2 \Re \sum_{t=0}^{L-2} \overline{a_t} a_{t+1} e^{\I(\alpha_t - \alpha_{t+1})} = 2 - 2 \sum_{t=0}^{L-2} r_t r_{t+1}.
$$
Then it suffices to optimize the amplitude distribution of the probe state; recall the normalization condition $\sum_{t=0}^{L-1} r_t^2 = 1$, and, without loss of generality, set $r_{-1} = r_L = 0$ as the boundary. Then it holds that
$$
\begin{aligned}
\inf_{ \mathbf{r} \in \mathbb{R}_{\geq 0}^{L}: \sum_{t = 0 }^{L - 1} r_t^2 = 1 } \E\left[C_{H}(\hat{\theta})\right] &= \inf_{ \mathbf{r} \in \mathbb{R}_{\geq 0}^{L}: \sum_{t = 0 }^{L - 1} r_t^2 = 1 } 2 - 2 \sum_{t=0}^{L-2} r_t r_{t+1} \\
&= \inf_{ \mathbf{r} \in \mathbb{R}_{\geq 0}^{L}: \sum_{t = 0 }^{L - 1} r_t^2 = 1 }  \sum_{t=-1}^{L-1} r_{t+1}^2 + \sum_{t=-1}^{L-1} r_{t}^2  - 2 \sum_{t=-1}^{L-1} r_t r_{t+1} \\
&= \inf_{ \mathbf{r} \in \mathbb{R}_{\geq 0}^{L}: \sum_{t = 0 }^{L - 1} r_t^2 = 1 } \sum_{t=-1}^{L-1} \left( r_{t+1} - r_t \right)^2.
\end{aligned}
$$
Suppose we set up the Hamiltonian $H_L$ with nearest-neighbor hopping 
\begin{equation}
\label{eqn:hopping_Hamiltonian}
H_L := 2 \openone - \sum_{t=0}^{L-2}\left( \ket{t+1} \bra{t} + \ket{t} \bra{t+1} \right),
\end{equation}
then the cost function can be reformulated as a matrix inner product
$$
\mathbf{r}^T H_L \mathbf{r} = \mathbf{r}^T \left( 2 r_t - r_{t-1} - r_{t+1} \right)_{t=0}^{L-1} = \sum_{t=-1}^{L-1} (r_{t+1} - r_t)^2.
$$
Therefore, it suffices to find the ground state of the Hamiltonian $H_L$ to minimize the cost. Notably, any eigenvector $\mathbf{r}$ of $H_L$ satisfies $H_L \mathbf{r} = \omega \mathbf{r}$, and at the entry level, we will need $2 r_t - r_{t-1} - r_{t+1} = \omega r_t$, or $r_{t+1} = (2 - \omega) r_t - r_{t-1}$. This matches exactly the recurrence relation that defines the Chebyshev polynomial in Definition \ref{def:chebyshev_poly}. To solve for $\omega$ and thus analytically determine $\mathbf{r}$, we will need the following lemma.

\begin{lemma}
\label{lemma:solving_recurrence_relation}
    Take $\{T_n\}_{n \geq - 1}$ as the series of Chebyshev polynomials of the second kind. If a non-zero sequence $\{r_t\}_{t=-1}^{L}$ satisfies $r_{-1} = 0$ and $r_{t+1} = 2x r_t - r_{t-1}$ for any $t\in \{0, 1, \dots, L-1\}$, then for any $-1 \leq t \leq L$, it holds that $r_t = r_0 T_t(x)$. Moreover, if $r_L = 0$, then it must hold that $T_L(x) = 0$.
\end{lemma}

\begin{proof}[Proof of \lref{lemma:solving_recurrence_relation}]
    We will prove the identity by induction. The base cases hold by the initial conditions that $T_{-1}(x) = 0$ gives $r_{-1} = 0 = r_0 T_{-1}(x)$, and $T_0(x) = 1$ gives $r_0 = r_0 T_0(x)$. Then for any $t \geq 1$, conditioned on the argument holds for $t$ and $t-1$, we have
    $$
    r_{t+1}(x)   = 2x r_{t}(x) - r_{t-1}(x) = r_0 \left( 2x T_t(x) - T_{t-1}(x) \right)= r_0 T_{t+1}(x).
    $$
    Since the sequence is non-zero, $r_0 \neq 0$, then $r_L = r_0 T_L(x)$ implies readily that $T_L(x) = 0$.
\end{proof}

By the statement of \lref{lemma:solving_recurrence_relation}, we can solve the explicit expression for $\{ r_t \}_{t=-1}^{L}$. For $T_L(x) = 0$, the trigonometric expression gives
$$
T_L(\cos \theta) = \frac{\sin((L + 1) \theta)}{\sin \theta} = 0 \quad \implies \quad \theta \in \{\theta\}_{j=1}^L, ~\theta_j = \frac{j \pi}{L+1}.
$$
So that in our case, the eigenvalues $\omega_j$ from tentative choice of $x_j = \cos \theta_j$ reads
\begin{equation} \label{eqn:eigenvectors_of_hopping_hamiltonian}
\begin{aligned}
&\omega_j = 2 - 2\cos \theta_j = 4 \sin^2\left(\frac{\theta_j}{2} \right) = 4 \sin^2 \left( \frac{j \pi}{2(L+1)} \right); \\
&r_t(j) = r_0 T_{t}(x_j) = \sin \theta_j \cdot \frac{\sin((t + 1) \theta_j)}{\sin \theta_j} = \sin((t + 1) \theta_j) = \sin \left( \frac{j(t+1) }{L+1} \pi \right),
\end{aligned}
\end{equation}
where the second equality holds for endpoint choices of $\theta_j$ via the extension [cf. Definition \ref{def:chebyshev_poly}]. Normalizing over $t$ for a fixed $j$, the discrete sine orthogonality identity \cite{doi:10.1137/S0036144598336745} gives $\sum_{t=0}^{L-1} r_t(j) r_t(j') =  \frac{L+1}{2} \delta_{j, j'}$, so that
$$
\forall \, 0 \leq t \leq L - 1,~1 \leq j \leq L, \quad r_t = \sqrt{\frac{2}{L+1}} \sin\left(  \frac{(t + 1) j}{L+1} \pi \right).
$$
Note that via the expression of $\omega_j$, the minimal value is attained at $j = 1$, so that taking the ground state as $\mathbf{r} = (r_t(1))_{t=0}^{L-1}$ recovers the optimal state for phase estimation derived in \cite{Bu_ek_1999, vanDam_2007}.

\paragraph{Extending to learning $\U(d)$ unitary channels.} Recall the reformulation result in Section~\ref{sec:problem_reformulation}. Roughly speaking, this reformulation tells us that, in our covariant learning framework, it suffices to ensure that the sampled unitaries have their spectra concentrated near the origin in a suitable sense with high probability. To design the probability distribution that generalizes the $\U(1)$ estimation theory, we need to redesign the cost function so that it jointly penalizes the deviation of each eigenvalue in $\spec(W)$. A natural choice that balances tractability against a small diamond distance is the additive Holevo cost:
\begin{equation} \label{eqn:additive_holevo_cost}
\forall \hat \bmt \in \T^d, \quad C_{H}(\hat \bmt) :=  4 \sum_{j=1}^d\sin^2 \left( \frac{\hat{\theta}_j}{2} \right) = 2d - \sum_{j=1}^d \left( e^{\I \hat{\theta}_j} + e^{-\I \hat{\theta}_j} \right).
\end{equation}
To track the property of the unitary channel, for notational consistency, we define $C_{H}(W) := C_{H}(\bmt(W))$. The likelihood function suggested in \eref{eqn:bisio_likelihood} naturally suggests an auxiliary probe in the antisymmetric subspace: For a fixed-$n$ strategy $(\mathsf{Y}, \bfq)$, take $L > d$ with $\Y \subseteq \Gamma_{L, d}$ [see Definition \ref{def:canonical_transformation}], by Lemma \ref{lemma:weyl_integral_formula} and Remark \ref{remark:reformulating_Weyl_character}, 
$$
\Pr[\bm{\theta}(W)\in A] = \int_{A} \left| \sum_{\lambda \in \Y} \sqrt{q_{\lambda}} \psi_{\bfk(\lambda)}(\bmt) \right|^2 \dd \nu^{\otimes d}(\bmt) = \int_{A} \left| \bra{\bmt}_{\wedge} \ket{\Phi_{n, \bfq}} \right|^2 \dd \nu^{\otimes d}(\bmt),~ \ket{\Phi_{n, \bfq}} = \sum_{\lambda \in \Y} \sqrt{q_{\lambda}} \ket{\bfk(\lambda)}_{\wedge} .
$$
for every permutation-invariant measurable set $A\subseteq\mathbb{T}^d$.
These are precisely the phase statistics corresponding to the normalized auxiliary probe state $\ket{\Phi_{n, \bfq}} \in \wedge^d \H_L$ [cf. \eref{eqn:canonical_probe_state}]. We therefore consider the following probe state and POVM on the full subspace $\wedge^d \H_L$ as a relaxation of the fixed-query construction in Section \ref{sec:bisio_learning_strategy}:
$$
\ket{\Psi_{L, d}} = \sum_{\bfk \in \Omega_{L, d} } a_{\bfk} \ket{\bfk}_{\wedge}: \quad \sum_{\bfk \in \Omega_{L, d} } |a_{\bfk}|^2 = 1; \quad M_{\wedge^d}(\dd \bm{\theta}) = \ket{\bmt}_{\wedge} \bra{\bmt}_{\wedge} \dd \nu^{\otimes d}(\bm{\theta}) \implies \int_{\T^d} M_{\wedge^d}(\dd \bm{\theta}) = \openone_{\wedge^d \H_L}.
$$
Transforming to the phase representation [cf. Definition \ref{def:Fourier_wavefunction}], for a chosen $\bm{\theta} \in \T^d$, its wavefunction yields
$$
\widehat{\Psi}_{L, d}(\bm{\theta}) = \bra{\bm{\theta}}_{\wedge} \ket{\Psi_{L, d} } = \sum_{\bfk \in \Omega_{L, d} } a_{\bfk} \bra{\bm{\theta}}_{\wedge} \ket{\bfk }_{\wedge} = \sum_{\bfk \in \Omega_{L, d} } a_{\bfk} \psi_{\bfk}(\bm{\theta}).
$$
Therefore, the expected cost is given by
\begin{equation}
\label{eqn:expected_additive_cost}
\begin{aligned}
    \E \left[ C_{H}(\hat \bmt) \right] = \int_{\T_d} C_{H}(\hat \bmt) \bra{\Psi_{L, d}} M_{\wedge^d}(\dd \hat \bmt) \ket{\Psi_{L, d}} =\int_{\T^d} C_{H}(\hat \bmt) \left| \sum_{\bfk \in \Omega_{L, d} } a_{\bfk} \psi_{\bfk}(\hat \bmt) \right|^2 \dd \nu^{\otimes d}(\hat \bmt).
\end{aligned}
\end{equation}
As before, we recast the minimization of the expected cost as a ground-state problem, while this time restricted to the antisymmetric subspace; the precise statement is the following lemma.

\begin{lemma}[Additive Holevo cost via a restricted local Hamiltonian] \label{lemma:expected_additive_cost_as_an_inner_product}
    Let $H_L$ be defined in \eref{eqn:hopping_Hamiltonian}. Define $H_{L,d}^{\wedge}$ as the restriction of the local Hamiltonian $\widetilde H_{L,d}$ to the antisymmetric subspace $\wedge^d\mathcal H_L$:
    $$
    H_{L, d}^{\wedge} = \widetilde{H}_{L, d} \bigg|_{\wedge^d \H_L}: \quad \widetilde{H}_{L, d} = \sum_{j=1}^d \openone^{\otimes (j - 1)} \otimes H_L \otimes \openone^{\otimes (d - j)},
    $$
    then the expected cost in \eref{eqn:expected_additive_cost} can be reformulated as
    $$
    \E \left[ C_{H}(\hat \bmt) \right] = \braket{ \Psi_{L, d} | H_{L, d}^{\wedge} |  \Psi_{L, d}}.
    $$
\end{lemma}

\begin{proof}[Proof of \lref{lemma:expected_additive_cost_as_an_inner_product}]
    We start by showing the action of the Hamiltonian $H_{L, d}^{\wedge}$ on exterior products. Firstly, note that the local Hamiltonian $\widetilde{H}_{L, d}$ commutes with any permutation $P_{\sigma}$ where $\sigma \in \S_d$. By linearity, we have $[\widetilde{H}_{L, d}, \Pi_{\wedge^d}] = 0$. Therefore, recall Definition \ref{def:anti_symm_ext_prod}, for any $\ket{v_1}, \ket{v_2}, \dots, \ket{v_d} \in \H_L$, we have
    \begin{equation} \label{eqn:action_rule_of_total_ham}
    \begin{aligned}
    H_{L, d}^{\wedge} \ket{v_1 \wedge v_2\wedge \cdots \wedge v_d} &= \sqrt{d!} \widetilde{H}_{L, d} \Pi_{\wedge^d}  \ket{v_1, v_2, \dots, v_d} = \sqrt{d!} \Pi_{\wedge^d} \widetilde{H}_{L, d} \ket{v_1, v_2, \dots, v_d} \\
    &= \sqrt{d!} \Pi_{\wedge^d} \sum_{j=1}^d \ket{v_1, \dots, H_L v_j, \dots, v_d} = \sum_{j=1}^d \ket{v_1 \wedge \cdots \wedge H_L v_j \wedge \cdots \wedge v_d}.
    \end{aligned}
    \end{equation}
    Then for any occupation bases $\bfk, \bfk' \in \Omega_{L, d}$, recall the action of $H_L$ [cf. \eref{eqn:hopping_Hamiltonian}], with the convention that $\ket{-1} = \ket{L} = 0$ [cf. the optimal $\U(1)$ learning derivation], the inner product reads
    \begin{equation} \label{eqn:basis_inner_product}
    \begin{aligned}
    \bra{\bfk'}_{\wedge} H_{L, d}^{\wedge} \ket{\bfk}_{\wedge} &= \bra{\bfk'}_{\wedge} \sum_{j=1}^{d} \ket{k_1 \wedge \cdots \wedge H_L k_j \wedge \cdots \wedge k_d} \\
    &= \bra{\bfk'}_{\wedge} \left(  2d \ket{\bfk}_{\wedge} - \sum_{j=1}^d \ket{k_1 \wedge \cdots \wedge (k_j - 1) \wedge \cdots \wedge k_d} - \sum_{j=1}^d \ket{k_1 \wedge \cdots \wedge (k_j + 1) \wedge \cdots \wedge k_d} \right) \\
    &= 2d \delta_{\bfk, \bfk'} - \sum_{j=1}^d \mathbf{1}_{\Omega_{L, d}}\left( \bfk + \mathbf{e}_j \right) \cdot \delta_{\bfk + \mathbf{e}_j, \bfk'} - \sum_{j=1}^d \mathbf{1}_{\Omega_{L, d}}\left( \bfk - \mathbf{e}_j \right) \cdot \delta_{\bfk - \mathbf{e}_j, \bfk'},
    \end{aligned}
    \end{equation}
    where by the indicator functions we ensure that $\bfk \pm \mathbf{e}_j$ keeps inside the legitimate support $\Omega_{L, d}$, and set them to $0$ via the boundary condition. Meanwhile, by extending the determinant formula for $\psi_{\bfk}$ to arbitrary integer tuples, we have the following pointwise identity
    $$
    \begin{aligned}
   &\sum_{j=1}^d e^{\pm \I \theta_j} \cdot \psi_{\bfk}(\bm{\theta}) = \frac{1}{\sqrt{d!}} \sum_{\ell=1}^d e^{\pm \I \theta_\ell} \det\left[ e^{\I k_j \theta_i} \right]_{1 \leq i, j \leq d} = \frac{1}{\sqrt{d!}} \sum_{\ell=1}^d \det\left[ e^{\I (k_j \pm \delta_{j, \ell}) \theta_i} \right]_{1 \leq i, j \leq d} =  \sum_{\ell=1}^d \psi_{\bfk \pm \mathbf{e}_\ell}(\bm{\theta}).
    \end{aligned}
    $$
    Here we have used the fact that a repeated frequency makes the determinant vanish identically [cf. Fact \ref{fact:alternating_property}], and a frequency dropping outside $\{0, \dots, L-1\}$ gives a nonzero function while being orthogonal to every basis $\psi_{\bfk'}$ with $\bfk' \in \Omega_{L,d}$ in $\mathbb{L}^2(\T^d, \dd \nu^{\otimes d})$. Such terms would therefore vanish in the formulation of the inner products below. Therefore, if we plug in the expression in \eref{eqn:additive_holevo_cost}, by the orthonormality [cf. Fact \ref{fact:wavefuntion_orthonormality}] and the basis expectation result in \eref{eqn:basis_inner_product}, the expected cost can be reformulated as
    \begin{align*}
        \E \left[ C_{H}(\hat \bmt) \right] &=  \int_{\T^d} C_{H}(\hat \bmt) \left| \sum_{\bfk \in \Omega_{L, d} } a_{\bfk} \psi_{\bfk}(\hat \bmt) \right|^2 \dd \nu^{\otimes d}(\hat \bmt) \\
        &=   \int_{\T^d}\left( 2d - \sum_{j=1}^d \left( e^{\I \hat{\theta}_j} + e^{-\I \hat{\theta}_j} \right) \right) \sum_{\bfk, \bfk' \in \Omega_{L, d} } \overline{a_{\bfk'}} a_{\bfk} \overline{\psi_{\bfk'}(\hat{\bm{\theta}})} \psi_{\bfk}(\hat{\bm{\theta}}) \dd \nu^{\otimes d}(\hat{\bm{\theta}}) \\
        &= \sum_{\bfk, \bfk' \in \Omega_{L, d} } \overline{a_{\bfk'}} a_{\bfk}\left(  2d \delta_{\bfk, \bfk' } -   \sum_{j=1}^d \int_{\T_d} \left(\overline{\psi_{\bfk'}(\hat \bmt)} \psi_{\bfk + \mathbf{e}_j}(\hat \bmt) + \overline{\psi_{\bfk'}(\hat \bmt)} \psi_{\bfk - \mathbf{e}_j}(\hat \bmt) \right) \dd \nu^{\otimes d}(\hat \bmt)  \right) \\
        &= \sum_{\bfk, \bfk' \in \Omega_{L, d}} \overline{a_{\bfk'}} a_{\bfk} \left(2d \delta_{\bfk, \bfk'} - \sum_{j=1}^d \mathbf{1}_{\Omega_{L, d}}\left( \bfk + \mathbf{e}_j \right) \cdot \delta_{\bfk + \mathbf{e}_j, \bfk'} - \sum_{j=1}^d \mathbf{1}_{\Omega_{L, d}}\left( \bfk - \mathbf{e}_j \right) \cdot \delta_{\bfk - \mathbf{e}_j, \bfk'} \right) \\
        &= \sum_{\bfk, \bfk' \in \Omega_{L, d}} \overline{a_{\bfk'}} a_{\bfk} \bra{\bfk'}_{\wedge} H_{L, d}^{\wedge} \ket{\bfk}_{\wedge} = \braket{ \Psi_{L, d} | H_{L, d}^{\wedge} |  \Psi_{L, d}}.
    \end{align*}
    This completes the proof.
\end{proof}

The elegant form obtained via \lref{lemma:expected_additive_cost_as_an_inner_product} allows us to minimize the additive Holevo cost simply by finding the ground state energy of $H_{L, d}^\wedge$. As the defining local Hamiltonian $\widetilde{H}_{L, d}$ consists of commuting terms, if we take the full set of eigenvectors $\{\mathbf{r}(j)\}_{j=1}^{L}$ with entries determined in \eref{eqn:eigenvectors_of_hopping_hamiltonian}, we have $H_L \mathbf{r}(j) = \omega_j \mathbf{r}(j)$ by \eref{eqn:action_rule_of_total_ham}, then for a strictly increasing index set $j_1 < j_2 < \cdots < j_d \in \{1, 2, \dots, L\}$, since
$$
\begin{aligned}
    H_{L, d}^{\wedge} \ket{\mathbf{r}(j_1) \wedge \cdots \wedge \mathbf{r}(j_d)} = \sum_{\ell=1}^d \ket{\mathbf{r}(j_1) \wedge \cdots \wedge H_L\mathbf{r}(j_\ell) \wedge \cdots \wedge \mathbf{r}(j_d)} = \left( \sum_{\ell = 1}^d \omega_{j_{\ell}} \right) \ket{\mathbf{r}(j_1)  \wedge \cdots \wedge \mathbf{r}(j_d)},
\end{aligned}
$$
an eigenvalue takes the form $\sum_{\ell}^d \omega_{j_{\ell}}$. Therefore, for the sake of the minimal eigenvalue, as for $d \leq L$, we have $\omega_j < \omega_{j+1}$ always for $j \leq L - 1$ [cf. \eref{eqn:eigenvectors_of_hopping_hamiltonian}], taking $j_{\ell} = \ell$ immediately yield the minimal expected cost. Correspondingly, the optimal probe state should take the form
\begin{equation}
\label{eqn:optimal_probe_state_for_additive_holevo_cost}
\begin{aligned}
\ket{\Psi_{L, d}} = \ket{\mathbf{r}(1) \wedge \cdots \wedge \mathbf{r}(d) } \implies a_{\bfk} &= \bra{\bfk}_{\wedge} \ket{\Psi_{L, d}} = \braket{k_1 \wedge \cdots \wedge k_d | \mathbf{r}(1)  \wedge \cdots \wedge \mathbf{r}(d)} \\
&= \det\left[ r_{k_j}(\ell) \right]_{1 \leq j, \ell \leq d} = \det\left[ \sqrt{\frac{2}{L+1}} \sin\left( \frac{(k_j + 1) \ell}{L+1} \pi \right) \right]_{1 \leq j, \ell \leq d}.
\end{aligned}
\end{equation}
We now show a standard fact about the optimal distribution $a_{\bfk}$.

\begin{fact} \label{fact:sign_of_optimal_distribution}
    For any $\bfk \in \Omega_{L, d}$, $|a_{\bfk}| = (-1)^{\binom{d}{2}} a_{\bfk}$.
\end{fact}

\begin{proof}[Proof of Fact \ref{fact:sign_of_optimal_distribution}]
    It suffices to track the sign of $a_{\bfk}$. Note that if we expand the expression of $a_{\bfk}$ using the trigonometric Vandermonde determinant identity \cite[Problem 350]{proskuryakov1974problems}, we have
    $$
    a_{\bfk} = \left( \frac{2}{L+1} \right)^{d/2} 2^{\binom{d}{2}} \prod_{i=1}^d \sin x_{i} \prod_{1 \leq j < \ell \leq d} \left( \cos x_{\ell} - \cos x_{j} \right): \quad x_{j} = \frac{k_{j} + 1}{L + 1} \pi.
    $$
    Since the index $k_j$ monotonically increases, we have $0 < x_j < x_{\ell} < \pi$ for $j < \ell$ and thus $\cos x_j > \cos x_{\ell}$, so each of the $\binom{d}{2}$ factors in $\prod_{1 \leq j < \ell \leq d}\left( \cos x_{\ell} - \cos x_{j} \right)$ is negative, thus contributing the sign $(-1)^{\binom{d}{2}}$.
\end{proof}

\paragraph{Fitting into the Bisio-Chiribella-D’Ariano-Facchini-Perinotti learning framework.} Note that the optimal probe state in the previous paragraph is dedicated to $\wedge^d \H_L$. To instantiate the learning strategy in our targeted framework, we must translate this auxiliary fermionic state back into the fixed-query primitive of the canonical learning framework. In particular, the occupation labels $\mathbf{k}\in\Omega_{L,d}$ are first partitioned according to the induced query number $n(\mathbf{k})=|\lambda(\mathbf{k})| = |\bfk| - \binom{d}{2}$, after which the squared amplitudes of the fermionic ground state determine the conditional representation weights within each fixed-$n$ sector.

Recall Definition \ref{def:canonical_transformation} and Corollary~\ref{cor:decomp_ext_space_with_box_number}, for a fixed choice of $\Omega_{L, d} = \bigcup_{n=0}^{Q} \Omega_{L, d}^{(n)}$ where $Q = d(L-d)$, we group the optimal probability distribution $\{|a_{\bfk}|^2\}_{\bfk \in \Omega_{L,d}}$ according to the number of boxes $n = n(\bfk)$ in the Young diagram associated with $\bfk$ [see \eref{eqn:optimal_probe_state_for_additive_holevo_cost}] to induce a probability distribution for $n$:
$$
\forall n \in \{0, 1, \dots, Q\}, \quad \mathsf{p}(n) := \sum_{\bfk\in \Omega_{L, d}^{(n)}  } |a_{\bfk}|^2.
$$
Then for a fixed $n$, we establish the probability distribution over Young diagrams and set up the probe state
\begin{equation}
\label{eqn:conditioned_working_probe_state}
\forall n: \mathsf{p}(n) > 0,~q_{\lambda(\bfk)}^{(n)} = \frac{|a_{\bfk}|^2}{\mathsf{p}(n)}, \quad \forall \bfk \in \Omega_{L, d}^{(n)} ~ \implies ~ \ket{\Psi_{n, \bfq^{(n)}}} = \bigoplus_{\bfk \in \Omega_{L, d}^{(n)} }  \frac{|a_{\bfk}|}{\sqrt{d_{\lambda(\bfk)}\mathsf{p}(n) }} \Ket{\openone_{\mathcal{W}_{\lambda(\bfk)}^d}} \otimes \ket{\xi_{\lambda(\bfk)}}.
\end{equation}
Note that including potentially zero-probability sectors leaves the probability formulas below unchanged, so for clarity we take $n$ to be supported on $\{0, 1, \dots, Q\}$. Our protocol then works as follows: we first sample $n \sim \mathsf{p}(n)$, and conditioned on $n$, we make $n$ parallel black-box queries to the unknown unitary $U$ and run the canonical learning framework to obtain the final estimate $\widehat{U}$.
Equivalently, we can pad extra registers and write the probe state and the POVM $\mathscr{M} = \{ \Psi_{\widehat{U}}  \dd \widehat{U} \}_{\widehat{U} \in \U(d)}$ from the global point of view
\begin{equation}
\label{eqn:probe_state_POVM_for_parallel_tomo}
\Psi_Q = \sum_{n=0}^{Q} \mathsf{p}(n) \ket{n} \bra{n} \otimes \ket{\Psi_{n, \bfq^{(n)}}} \bra{\Psi_{n, \bfq^{(n)}}} \otimes \ket{0^{Q-n}} \bra{0^{Q -n}}, ~ \Psi_{\widehat{U}} = \sum_{n=0}^{Q} \ket{n} \bra{n} \otimes \ket{\Psi_{n, \widehat{U}}} \bra{\Psi_{n, \widehat{U}}}  \otimes \openone^{\otimes (Q-n)}.
\end{equation}
Here since
$\Tr[ \mathcal{U}^{\otimes (Q-n)} (\ket{0^{Q-n}} \bra{0^{Q-n}}) ] = 1$, the extra queries applied to the idle registers do not affect the outcome distribution.
Note that by this design, we have discarded the coherence between sectors with different numbers of total queries. Fortunately, as we will show in our subsequent analysis, for the purpose of minimizing the diamond distance under covariant protocol\footnote{
Note that this is already a weaker requirement than ensuring a small additive Holevo cost: The diamond distance concerns the pairwise arc distance between eigenphases and is insensitive to a joint shift; see Corollary \ref{corollary:success_tomo_event}.
}, it suffices to exploit such classical correlation.

\section{Proof of \tref{thm:query_optimal_parallel_tomo}}
\label{sec:proof_of_main_theorem}

To prove the main theorem, we first establish a general bound on the error probability using properties of the diamond distance and of covariant estimation. We then instantiate this bound with the construction of
Section~\ref{sec:contructing_learning_algorithm} to obtain the claimed query
complexity.

\subsection{Reformulating the success probability}
While the diamond distance is determined entirely by the relative geometry of the eigenphases of $W$, its explicit analytical form does not lend itself to direct estimation. We therefore need the following lemma to reformulate the target success event in a more convenient form.

\begin{lemma}[{\cite[Section 1.1]{HKOT23}}, adapted]
\label{lemma:reformulating_diamond_norm}
Let $V\in\U(d)$ and suppose $\spec(V) = \{e^{\I\theta_j}\}_{j=1}^d$, counted with multiplicity. Define
$$
R_{\sp}(V):=\inf_{\beta\in\mathbb R}\max_{j\in[d]}\left|\Arg\!\left(e^{\I(\theta_j-\beta)}\right)\right|=\inf_{\beta\in\mathbb R}\max_{j\in[d]}\min_{k\in\mathbb Z}|\theta_j-\beta+2\pi k|.
$$
Then
$$
\left\|\mathcal V-\mathcal I\right\|_\diamond=2\sin\!\left(\min\left\{R_{\sp}(V),\frac\pi2\right\}\right).
$$
\end{lemma}

\begin{proof}[Proof of \lref{lemma:reformulating_diamond_norm}]
By \cite[Proposition 1.6]{HKOT23},
$$
\left\|\mathcal V-\mathcal I\right\|_\diamond= 2\sqrt{1-\kappa(V)^2}, \quad \kappa(V):=\min_{z\in\conv(\spec(V))}|z|.
$$
It remains to show that $\kappa(V)=\max\{0,\cos R_{\sp}(V)\}$. The infimum in the definition of $R_{\sp}(V)$ is attained. First, we suppose that $R_{\sp}(V)\leq\pi/2$. After a common rotation, we choose extreme points representatives $\phi_j\in[-R_{\sp}(V),R_{\sp}(V)]$ of the eigenphases. For every $z=\sum_jp_je^{\I\phi_j}$ in their convex hull, we have
$$
|z|\geq\Re z=\sum_jp_j\cos\phi_j\geq\cos R_{\sp}(V).
$$
By minimality of the covering arc, both endpoints are spectral points when $R_{\sp}(V)>0$. The case $R_{\sp}(V)=0$ is immediate. Hence, the midpoint of the two endpoint eigenvalues must belong to the convex hull and has modulus $\cos R_{\sp}(V)$. Thus we have $\kappa(V)=\cos R_{\sp}(V)$ for this case.

Now suppose that $R_{\sp}(V)>\pi/2$. If $0\notin\conv(\spec(V))$, conventional separation theorem \cite{kakutani1937beweis} in the real plane gives $\beta\in\mathbb R$ and $c>0$ such that $\Re(e^{-\I\beta}z)\geq c$ for every $z\in\conv(\spec(V))$. In particular, every eigenphase lies in the open semicircle centered at $\beta$, so $|\Arg(e^{\I(\theta_j-\beta)})|<\pi/2$ for all $j$. This contradicts $R_{\sp}(V)>\pi/2$. Therefore, $0$ belongs to the convex hull and $\kappa(V)=0$, completing the proof.
\end{proof}

We can therefore reformulate the target success event:

\begin{corollary} \label{corollary:success_tomo_event}
    For $0<\eps<2$, Lemma \ref{lemma:reformulating_diamond_norm} implies that the success probability can be reformulated as
    $$
    \begin{aligned}
    \Pr\left[ \left\| \mathcal{W} - \mathcal{I} \right\|_{\diamond} \leq \eps \right] &= \Pr\left[  R_{\sp}(W) \leq \arcsin \frac{\eps}{2} \right] =: \Pr\left[ W \in \mathsf{E}_{\eps} \right].
    \end{aligned}
    $$
For $\bm\theta\in\mathbb T^d$, with slight notational overload\footnote{
We borrow this terminology from computer programming to allow the same function to accept different input types.
}, we write $R_{\sp}(\bm\theta):=R_{\sp}(\mathrm{diag}(e^{\I\theta_1},\ldots,e^{\I\theta_d}))$, and the event $W \in \mathsf{E}_{\eps}$ is equivalent to $\bmt(W) \in \widetilde{\mathsf{E}}_{\eps}$, where $\widetilde{\mathsf{E}}_{\eps} := \left\{ \bm{\theta} \in \T^d:~R_{\sp}(\bmt) \leq \arcsin \frac{\eps}{2}  \right\}$. Note that the spectral covering radius $R_{\sp}$ depends only on the unordered eigenphases and is invariant under a common rotation. Thus, for every permutation $\sigma\in\mathfrak S_d$ and every $t\in\mathbb R$, with $\sigma(\bm\theta):=(\theta_{\sigma(1)},\ldots,\theta_{\sigma(d)})$,
$$
R_{\sp}(\sigma(\bm\theta))=R_{\sp}(\bm\theta)=R_{\sp}(\bm \theta+ (t^d)),
$$
where the phase additions are all conducted modulo $2\pi$.
\end{corollary}

To start with, we reformulate the probability of obtaining the rotated estimate $W$ when the probe state \eref{eqn:probe_state_POVM_for_parallel_tomo} is used via the Born rule:
\begin{equation}
\label{eqn:likelihood_probability_reformulation_updated}
\begin{aligned}
    \Pr\left[ W  \right] &  = \Tr\left[ \Psi_{\widehat{U}} \cdot \mathcal{U}^{\otimes Q}(\Psi_{Q}) \right] = \sum_{n=0}^Q \mathsf{p}(n) \left| \braket{\Psi_{n, \widehat{U}} | U^{\otimes n} | \Psi_{n, \bfq^{(n)}}} \right|^2 = \sum_{n=0}^{Q}\left|  \sum_{\bfk \in \Omega_{L, d}^{(n)}} |a_{\bfk}| \chi_{\lambda(\bfk)}(W) \right|^2.
\end{aligned}
\end{equation}
Then, by Corollary \ref{corollary:success_tomo_event}, the indicator function $\mathbf{1}_{\mathsf{E}_{\eps}}$ is a class function on $\U(d)$, since it depends on any $W$ only through the geometry of its eigenphases counted with multiplicity and is invariant under their permutations. Consequently, the success probability reads
$$
\begin{aligned}
    \Pr\left[ W \in \mathsf{E}_{\eps} \right] &= \sum_{n=0}^Q \int_{\U(d)}\left| \sum_{\bfk \in \Omega_{L, d}^{(n)}} |a_{\bfk}| \chi_{\lambda(\bfk)}(W) \right|^2 \mathbf{1}_{\mathsf{E}_{\eps}}\left( W\right)\dd W \\
    \overset{(\ast)}&{=} \sum_{n=0}^Q\int_{\U(d)}\left| \sum_{\bfk \in \Omega_{L, d}^{(n)}} |a_{\bfk}| \cdot \frac{\det\left[ e^{\I \theta_j (d + \lambda_{\ell}(\bfk) - \ell) } \right]_{1 \leq j, \ell \leq d} }{ \prod_{1 \leq j < \ell \leq d} \left( e^{\I \theta_j} - e^{\I \theta_{\ell}} \right) } \right|^2 \mathbf{1}_{\mathsf{E}_{\eps}}(W) \dd W \\
    \overset{(\ast \ast)}&{=}  \frac{1}{d!}\sum_{n = 0}^Q \int_{\T^d} \left| \sum_{\bfk \in \Omega_{L, d}^{(n)}}  |a_{\bfk}| \cdot \det\left[ e^{\I \theta_j k_{d - \ell + 1} } \right]_{1 \leq j, \ell \leq d}   \right|^2 \mathbf{1}_{\widetilde{\mathsf{E}}_{\eps}}(\bmt) \dd \nu^{\otimes d} (\bm{\theta}).
\end{aligned}
$$
Here $(\ast)$ follows from Lemma \ref{lemma:weyl_char_formula} and the relation $\lambda_{\ell}(\bfk) + d - \ell = k_{d - \ell + 1}$ in the occupation basis-partition transformation [cf. Definition \ref{def:canonical_transformation}], and $(\ast \ast)$ is due to Lemma \ref{lemma:weyl_integral_formula}. By Remark~\ref{remark:reformulating_Weyl_character} and Fact~\ref{fact:sign_of_optimal_distribution}, and by reversing the column order of the matrix determinant at the cost of an extra sign factor $(-1)^{\binom{d}{2}}$, we can further express the probability in terms of the wavefunctions as
\begin{equation}
\label{eqn:probability_with_phase_integral}
\begin{aligned}
\Pr\left[ W \in \mathsf{E}_{\eps} \right] = \Pr\left[ \bmt(W) \in \widetilde{\mathsf{E}}_{\eps} \right] &= \frac{1}{d!} \sum_{n=0}^Q \int_{\widetilde{\mathsf{E}}_{\eps}}  \left| \sum_{\bfk \in \Omega_{L, d}^{(n)}}  |a_{\bfk}| \cdot (-1)^{\binom{d}{2}} \det\left[ e^{\I \theta_j k_{\ell} } \right]_{1 \leq j, \ell \leq d}   \right|^2 \dd \nu^{\otimes d} (\bm{\theta}) \\ &=  \sum_{n=0}^Q \int_{\widetilde{\mathsf{E}}_{\eps}} \left| \sum_{\bfk \in \Omega_{L, d}^{(n)} } a_{\bfk} \psi_{\bfk}(\bm{\theta}) \right|^2 \dd \nu^{\otimes d} (\bm{\theta}).
\end{aligned}
\end{equation}
To proceed, we use the following lemma to reformulate the integral.

\begin{lemma}
\label{lemma:vanishing_of_cross_term}
    Define the phase-dependent identity $F_n(\bm{\theta}) =  \sum_{\bfk \in \Omega_{L, d}^{(n)} } a_{\bfk} \psi_{\bfk}(\bm{\theta}) $. Then for any $m, n \in \{0, \dots, Q\}$,
    $$
    \int_{\widetilde{\mathsf{E}}_{\eps}} \overline{F_{m}(\bm{\theta})} F_{n}(\bm{\theta}) \dd \nu^{\otimes d}(\bm{\theta}) = \delta_{m, n} \int_{\widetilde{\mathsf{E}}_{\eps}} \left|F_{n}(\bm{\theta})\right|^2 \dd \nu^{\otimes d}(\bm{\theta}).
    $$
\end{lemma}

\begin{proof}[Proof of \lref{lemma:vanishing_of_cross_term}]
    As remarked in Corollary \ref{corollary:success_tomo_event}, the set $\widetilde{\mathsf{E}}_{\eps}$ is closed under any joint shift $\bm{\theta} \mapsto \bmt + (t^d)$, and
    $$
    \begin{aligned}
    F_n(\bmt + (t^d)) &= \sum_{ \bfk \in \Omega_{L,d}^{(n)} } a_{\bfk} \psi_{\bfk}(\bm{\theta} + (t^d)) = \sqrt{d!} \sum_{\bfk \in \Omega_{L, d}^{(n)} } a_{\bfk} \det\left[ e^{\I k_j (\theta_i + t)} \right]_{1 \leq i, j \leq d} \\
    &= \sqrt{d!} e^{\I t |\bfk|} \sum_{\bfk \in \Omega_{L, d}^{(n)} }a_{\bfk} \det\left[ e^{\I k_j \theta_i} \right]_{1 \leq i, j \leq d} = e^{\I t \left(n + \binom{d}{2} \right)}  F_n(\bmt).
    \end{aligned}
    $$
    Therefore, by the closure property, the integral over the target region yields
    $$
    \begin{aligned}
        \int_{\widetilde{\mathsf{E}}_{\eps}} \overline{F_{m}(\bm{\theta})} F_{n}(\bm{\theta}) \dd \nu^{\otimes d}(\bm{\theta}) &= \int_{\widetilde{\mathsf{E}}_{\eps}} \overline{F_{m}(\bm{\theta} + (t^d))} F_{n}(\bm{\theta} + (t^d)) \dd \nu^{\otimes d}(\bm{\theta} + (t^d)) \\
        &= e^{-\I t \left( m + \binom{d}{2} \right) } e^{\I t \left( n + \binom{d}{2} \right) }\int_{\widetilde{\mathsf{E}}_{\eps}} \overline{F_{m}(\bm{\theta})} F_{n}(\bm{\theta}) \dd \nu^{\otimes d}(\bm{\theta}) \\
        &= e^{\I t(n - m)} \int_{\widetilde{\mathsf{E}}_{\eps}} \overline{F_{m}(\bm{\theta})} F_{n}(\bm{\theta}) \dd \nu^{\otimes d}(\bm{\theta}).
    \end{aligned}
    $$
    Since this holds for any $t \in \R$, when $m \neq n$, setting $t = \frac{\pi}{n - m}$ gives rise to $\int_{\widetilde{\mathsf{E}}_{\eps}} \overline{F_{m}(\bm{\theta})} F_{n}(\bm{\theta}) \dd \nu^{\otimes d}(\bm{\theta}) = - \int_{\widetilde{\mathsf{E}}_{\eps}} \overline{F_{m}(\bm{\theta})} F_{n}(\bm{\theta}) \dd \nu^{\otimes d}(\bm{\theta})$, and thus $\int_{\widetilde{\mathsf{E}}_{\eps}} \overline{F_{m}(\bm{\theta})} F_{n}(\bm{\theta}) \dd \nu^{\otimes d}(\bm{\theta}) = 0$. This completes the proof.
\end{proof}
This cross-term elimination that follows from the translation invariance of the success event would substantially simplify the sum on the right-hand side of \eref{eqn:probability_with_phase_integral}:
\begin{equation}
\label{eqn:probability_of_successful_event_reformulation_updated}
\begin{aligned}
   \Pr\left[ \bmt(W) \in \widetilde{\mathsf{E}}_{\eps} \right] &= \sum_{n=0}^{Q} \int_{\widetilde{\mathsf{E}}_{\eps}} \left| \sum_{\bfk \in \Omega_{L, d}^{(n)} } a_{\bfk} \psi_{\bfk}(\bm{\theta}) \right|^2 \dd \nu^{\otimes d} (\bm{\theta}) \\
   &= \sum_{n=0}^{Q} \int_{\widetilde{\mathsf{E}}_{\eps}} \left| F_{n}(\bm{\theta}) \right|^2 \dd \nu^{\otimes d} (\bm{\theta}) + \sum_{m \neq n} \int_{\widetilde{\mathsf{E}}_{\eps}} \overline{F_{m}(\bm{\theta})} F_{n}(\bm{\theta}) \dd \nu^{\otimes d} (\bm{\theta}) \\
    &= \int_{\widetilde{\mathsf{E}}_{\eps}} \left| \sum_{n=0}^{Q} \sum_{\bfk \in \Omega_{L, d}^{(n)}} a_{\bfk} \psi_{\bfk}(\bm{\theta}) \right|^2 \dd \nu^{\otimes d} (\bm{\theta}) = \int_{\widetilde{\mathsf{E}}_{\eps}} \left| \sum_{\bfk \in \Omega_{L, d} } a_{\bfk} \psi_{\bfk}(\bm{\theta}) \right|^2 \dd \nu^{\otimes d} (\bm{\theta}).
\end{aligned}
\end{equation}
Define $\mu(\bmt) = | \sum_{\bfk \in \Omega_{L, d} } a_{\bfk} \psi_{\bfk}(\bm{\theta}) |^2$, by the orthogonality guarantee in Fact \ref{fact:wavefuntion_orthonormality} and the normalization condition $\sum_{\bfk \in \Omega_{L, d} } |a_{\bfk}|^2  = \sum_{n=0}^{Q} \mathsf{p}(n) = 1$, $\mu(\bmt)$ is a legitimate probability density over $\T^d$ with respect to $ \dd \nu^{\otimes d}(\bmt)$. Therefore, with slight notational overload, we henceforth let $\bm{\theta}$ denote an auxiliary random variable distributed according to $\mu(\bm{\theta}) \dd\nu^{\otimes d}(\bm{\theta})$.
By the preceding calculation, the probability that the eigenphases of $W$ produced by our learning protocol fall into $\widetilde{\mathsf{E}}_{\eps}$ coincides with $\Pr[\bmt \in \widetilde{\mathsf{E}}_{\eps}]$. We can henceforth replace $\bmt(W)$ with the auxiliary variable $\bmt$ of explicit density in the success probability and further lower bound it by passing to a smaller but more tractable event, via the following observation:

\begin{observation} \label{obs:lower_bound_for_probability}
    Set $I_{\eps} := \{ \theta \in \T: |\Arg(e^{\I \theta})| \leq \arcsin(\eps / 2) \}$, identified with $[-\arcsin(\eps / 2), \arcsin(\eps / 2)] \subset (-\pi, \pi]$. Then for any $\bm{\theta} \in I_{\eps}^{\times d}$, choosing $\beta = 0$ in the definition of $R_{\sp}$ [cf. Lemma~\ref{lemma:reformulating_diamond_norm}] gives $R_{\sp}(\bm{\theta}) \leq \arcsin(\eps/2)$, and hence $I_{\eps}^{\times d} \subseteq \widetilde{\mathsf{E}}_{\eps}$ [cf. Corollary \ref{corollary:success_tomo_event}]. Therefore, for any random variable $\bm{\theta}$ on $\mathbb{T}^d$,
    \begin{equation} \label{eqn:prob_lower_bound_by_domain_subset}
    \Pr\left[ \bmt  \in \widetilde{\mathsf{E}}_{\eps} \right] \geq \Pr\left[ \bmt  \in I_{\eps}^{\times d} \right].
    \end{equation}
\end{observation}
With the above observation, let $N(\bmt) = \sum_{j=1}^d \mathbf{1}_{I_{\eps}^c}(\theta_j)$ count the total number of coordinates of $\bmt$ falling outside $I_{\eps}$. Note that the right-hand side of \eref{eqn:prob_lower_bound_by_domain_subset} evaluates precisely to the probability of getting $N(\bmt) = 0$. By an argument analogous to that in \cite[Remark 1.11]{HKOT23}, Markov's inequality implies that
    \begin{equation}
    \label{eqn:relaed_probability_lower_bound}
    \begin{aligned}
    \Pr\left[ \bmt  \in I_{\eps}^{\times d} \right] = \Pr\left[ N(\bmt ) = 0 \right] &= 1 - \Pr\left[ N(\bmt) \geq 1 \right] \\
    &\geq 1 - \E\left[N(\bmt) \right]  = 1 - \sum_{j=1}^d\Pr\left[ \theta_j \in I_{\eps}^c \right].
    \end{aligned}
    \end{equation}

To sum up, it suffices to bound the tail probability for each $\theta_j$ to fall outside the interval of interest by combining the inequalities presented in \eref{eqn:prob_lower_bound_by_domain_subset} and \eref{eqn:relaed_probability_lower_bound}:
\begin{equation}
\label{eqn:expectation_of_total_count}
\Pr\left[ \bmt  \in \widetilde{\mathsf{E}}_{\eps} \right]\ge 1- \sum_{j=1}^d\Pr\left[ \theta_j  \in I_{\eps}^c \right].
\end{equation}

\subsection{Bounding the residual terms}

\label{subsec:tail_probability_bounding}

We now bound the residual terms $\E[N(\bmt)] = \sum_{j=1}^d \Pr[\theta_j \in I_{\eps}^c]$ under the auxiliary density. For each event $\theta_j \in I_{\eps}^c$, it suffices to expand its probability in terms of the marginal of a single coordinate $\theta_j$:
$$
    \sum_{j=1}^d \Pr\left[\theta_j \in I_{\eps}^c\right] = \sum_{j=1}^d \int_{I_{\eps}^c} \dd\nu(\theta_j) \int_{\T^{d-1}} \left| \sum_{\bfk \in \Omega_{L,d}} a_{\bfk} \psi_{\bfk}(\bmt) \right|^2 \dd\nu^{\otimes (d-1)}(\bmt_{\bar{j}}),
$$
where $\bmt_{\bar{j}} = (\theta_k)_{k \in [d] \setminus \{j\}}$ collects the random variables in the remaining $d-1$ coordinates. Recall that both $a_{\bfk}$ and $\psi_{\bfk}(\bm{\theta})$ have determinant forms [cf. Definition \ref{remark:reformulating_Weyl_character} and \eref{eqn:optimal_probe_state_for_additive_holevo_cost}], we can obtain a compact reformulation:
$$
\begin{aligned}
    \sum_{\bfk \in \Omega_{L, d} } a_{\bfk} \psi_{\bfk}(\bm{\theta}) &= \frac{1}{\sqrt{d!}}\sum_{0 \leq k_1 < \cdots < k_d \leq L - 1} \det\left[ e^{\I k_{\ell} \theta_j}  \right]_{1 \leq j, \ell \leq d}\det\left[ \sqrt{\frac{2}{L+1}} \sin\left( \frac{(k_j + 1) \ell}{L + 1} \pi \right)  \right]_{1 \leq j , \ell \leq d}   \\
    \overset{(\ast)}&{=} \frac{1}{\sqrt{d!}}\det\left[ \left( e^{\I t \theta_j} \right)_{1 \leq j \leq d, 0 \leq t \leq L-1}\left(\sqrt{\frac{2}{L+1}} \sin\left( \frac{(t + 1) \ell}{L + 1} \pi \right) \right)_{0 \leq t \leq L-1, 1 \leq \ell \leq d} \right] \\
    &= \frac{1}{\sqrt{d!}}\det\left[ \sqrt{\frac{2}{L+1}} \sum_{t=0}^{L-1}  e^{\I t \theta_j}\sin \left( \frac{(t + 1) \ell}{L+1} \pi \right) \right]_{1 \leq j, \ell \leq d} \\
    &:= \frac{1}{\sqrt{d!}}\det\left[ \varphi_{\ell}(\theta_{j}) \right]_{1 \leq j, \ell \leq d}.
\end{aligned}
$$
Here $(\ast)$ follows from the Cauchy-Binet formula [cf. Lemma \ref{lemma:cauchy_binet}].
We therefore arrive at the following exact expression for the residual terms of interest
\begin{equation}
\label{eqn:success_probability_expression}
    \sum_{j=1}^d\Pr\left[ \theta_j \in I_{\eps}^c \right] = \frac{1}{d!}\sum_{j=1}^d \int_{I_{\eps}^c} \dd \nu(\theta_j) \int_{\T^{d-1}} \left| \det\left[ \varphi_{\ell}(\theta_{k}) \right]_{1 \leq k, \ell \leq d}  \right|^2 \dd \nu^{\otimes (d-1)} (\bm{\theta}_{\overline{j}}),
\end{equation}
Observing the strong symmetry among all coordinates of $\bmt$, we can simplify the above integral to
\begin{equation}
\label{eqn:expectation_of_total_count1}
\sum_{j=1}^d\Pr\left[ \theta_j \in I_{\eps}^c \right] = \frac{1}{(d-1)!}\int_{I_{\eps}^c} \dd \nu(\theta_1) \int_{\T^{d-1}} \left| \det\left[ \varphi_{\ell}(\theta_{k}) \right]_{1 \leq k, \ell \leq d}  \right|^2 \dd \nu^{\otimes (d-1)} (\bm{\theta}_{\overline{1} }),
\end{equation}
Explicitly evaluating the probability via the marginal of $\theta_1$ gives
\begin{equation}
\label{eqn:expanding_by_marginal}
\begin{aligned}
    &\sum_{j=1}^d\Pr\left[ \theta_j \in I_{\eps}^c \right]  \\
    \overset{(\ast)}&{=} \frac{1}{(d-1)!} \int_{I_{\eps}^c} \dd \nu(\theta_1) \int_{\T^{d-1}} \left| \sum_{\ell = 1}^d (-1)^{\ell + 1} \varphi_{\ell}(\theta_1) \det\left[ \varphi_{r}(\theta_t) \right]_{2 \leq t \leq d, r \in [d] \setminus \{\ell\} } \right|^2 \dd \nu^{\otimes (d-1)}(\bm{\theta}_{\overline{1}}) \\
    \overset{}&{=} \frac{1}{(d-1)!} \int_{I_{\eps}^c} \dd \nu(\theta_1) \sum_{\ell, \ell' = 1}^d (-1)^{\ell + \ell'} \overline{\varphi_{\ell}(\theta_1)} \varphi_{\ell'}(\theta_1) \int_{\T^{d-1}} \overline{\det\left[ \varphi_{r}(\theta_t) \right]_{\substack{2 \leq t \leq d \\ r \in [d] \setminus \{\ell\} }}} \det\left[ \varphi_{r'}(\theta_{t'}) \right]_{\substack{2 \leq t' \leq d \\r' \in [d] \setminus \{\ell'\} }} \dd \nu^{\otimes (d-1)}(\bm{\theta}_{\overline{1}}) \\
    \overset{(\ast \ast)}&{=} \frac{1}{(d-1)!} \int_{I_{\eps}^c} \dd \nu(\theta_1) \sum_{\ell, \ell' = 1}^d  \overline{\varphi_{\ell}(\theta_1)} \varphi_{\ell'}(\theta_1) \cdot (d-1)! \cdot \delta_{[d] \setminus \{\ell\}, [d] \setminus \{\ell'\} } = \int_{I_{\eps}^c} \sum_{\ell=1}^d \left| \varphi_{\ell} (\theta) \right|^2 \dd \nu(\theta),
\end{aligned}
\end{equation}
where $(\ast)$ follows from the Laplace minor expansion in Corollary \ref{corollary:laplace_expansion}, the proof of the equality marked $(\ast\ast)$ is deferred to the end of this section to avoid interrupting the argument, and we have omitted the subscript in the last expression for brevity.
Further expanding the integrand yields
$$
\begin{aligned}
    \sum_{\ell=1}^d \left| \varphi_{\ell} (\theta) \right|^2 &= \frac{2}{L+1}\sum_{\ell = 1}^d \left| \sum_{t=0}^{L-1} e^{\I t \theta} \sin \left( \frac{(t + 1) \ell}{L+1} \pi \right) \right|^2 \\
    &= \frac{2}{L+1} \sum_{\ell=1}^d \left| \frac{1}{2\I} \left(  e^{\I \frac{\ell \pi}{L+1} } \frac{1 - (e^{\I (\theta + \frac{\ell \pi}{L+1})})^L }{1 - e^{\I (\theta + \frac{\ell \pi}{L+1})}} - e^{-\I \frac{\ell \pi}{L+1} } \frac{1 - (e^{\I (\theta - \frac{\ell \pi}{L+1})})^L }{1 - e^{\I (\theta - \frac{\ell \pi}{L+1})}} \right) \right|^2
    \\
    &= \frac{2}{L+1} \sum_{\ell=1}^d \left| \frac{ \sin \left( \frac{\ell \pi}{L+1} \right) \left( 1 - (-1)^{\ell} e^{\I (L + 1) \theta} \right) }{1 - 2 e^{\I \theta } \cos\left(\frac{\ell \pi}{L + 1}\right) + e^{2 \I \theta } }  \right|^2
    \\
    &= \frac{1}{L+1} \sum_{\ell=1}^d \frac{ \sin^2 \left(\frac{\ell \pi}{L+1} \right) \left( 1 - (-1)^\ell \cos((L + 1)\theta) \right) }{\left( \cos \theta - \cos \frac{\ell \pi}{L+1} \right)^2}.
\end{aligned}
$$
And at points $\theta = \pm\frac{\ell\pi}{L+1}$ we interpret the quotient by continuity, while our choice of $L$ below keeps these points outside the tail region
$I_\varepsilon^c$. Indeed, for $\eps$ such that $0 < \eps < 2$, choosing $L = \lceil \frac{8 \pi d}{\eps} \rceil - 1$ suffices to guarantee that $L > d$ and $\frac{\eps}{2} \geq \frac{4 \pi d }{L + 1} $. Take the auxiliary random variable $\vartheta = \vartheta(\theta) := | \Arg(e^{\I \theta}) | \in [0, \pi]$ in a single period. As $\arcsin x \geq x$ for $x \in [0, 1]$, we can derive a lower bound for $\vartheta$ when $\theta \in I_{\eps}^c$: For every $\ell \in [d]$, it holds that
$$
\vartheta > \arcsin \frac{\eps}{2} \geq \frac{\eps}{2} \geq \frac{4\pi d}{L + 1} \geq \frac{4 \pi \ell}{L + 1}.
$$
Consequently, using the handy inequalities $\frac{x}{\pi} \leq \sin \frac{x}{2} \leq \frac{x}{2}$ for $x \in [0, \pi]$, and noting that $\frac{\ell \pi}{2(L+1)} \leq \frac{\vartheta}{8} < \frac{\vartheta}{4}$, the denominator satisfies
$$
\cos \frac{\ell \pi}{L+1} - \cos \vartheta = 2 \left( \sin^2 \frac{\vartheta}{2} - \sin^2 \frac{\ell \pi}{2(L+1)} \right) > 2\left( \frac{1}{\pi^2} - \frac{1}{16} \right) \vartheta^2.
$$
Meanwhile, the numerator can be simply bounded from above by
$$
\sin^2 \left(\frac{\ell \pi}{L+1} \right) \left( 1 - (-1)^\ell \cos((L + 1)\theta) \right) \leq 2 \sin^2 \left(\frac{\ell \pi}{L+1} \right) \leq \frac{2\ell^2 \pi^2}{(L + 1)^2}.
$$
And therefore
$$
\begin{aligned}
    \sum_{\ell=1}^d \left| \varphi_{\ell} (\theta) \right|^2 < \frac{1}{L+ 1} \sum_{\ell = 1}^d \frac{\frac{2\ell^2 \pi^2}{(L + 1)^2}}{\left(2\left( \frac{1}{\pi^2} - \frac{1}{16} \right) \vartheta^2 \right)^2} = \frac{\pi^2 \sum_{\ell = 1}^d \ell^2}{2 \left(\frac{1}{\pi^2} - \frac{1}{16}  \right)^{2} (L+1)^3 }  \frac{1}{\vartheta^4} = \underbrace{\frac{\pi^2  }{12\left(\frac{1}{\pi^2} - \frac{1}{16}  \right)^{2} }}_{=: C} \frac{d(d + 1)(2d + 1)}{(L+1)^3} \frac{1}{\vartheta^4}.
\end{aligned}
$$
 By the symmetry on the phase circle and noting that $\dd \nu(\theta) = \dd \vartheta / 2\pi$, it suffices to track one of the branches, so that \eref{eqn:expanding_by_marginal} is bounded by
$$
\begin{aligned}
    \sum_{j=1}^d\Pr\left[ \theta_j \in I_{\eps}^c \right]   = \int_{I_{\eps}^c} \sum_{\ell=1}^d \left| \varphi_{\ell} (\theta) \right|^2 \dd \nu(\theta) &< C \cdot \frac{d(d + 1)(2d + 1)}{(L+1)^3} \frac{1}{\pi} \int^{\pi}_{\arcsin(\eps/2)} \frac{\dd \vartheta}{\vartheta^4} \\
    &\leq C \cdot \frac{d(d + 1)(2d + 1)}{(L+1)^3} \frac{1}{\pi} \int^{\pi}_{\eps/2} \frac{\dd \vartheta}{\vartheta^4} \\
    &= C \cdot \frac{d(d + 1)(2d + 1)}{3\pi (L+1)^3} \left( \frac{8}{\eps^3} - \frac{1}{\pi^3} \right) < C \cdot \frac{8d(d + 1)(2d + 1)}{3\pi (L+1)^3 \eps^3}.
\end{aligned}
$$
Finally, combining the actual-to-auxiliary identity in \eref{eqn:probability_of_successful_event_reformulation_updated}, the reformulation in Observation \ref{obs:lower_bound_for_probability}, and the preceding inequality into \eref{eqn:expectation_of_total_count}, we finally get
$$
\begin{aligned}
\Pr\left[ \left\| \widehat{\mathcal{U}} - \mathcal{U} \right\|_{\diamond}\leq\eps \right] = \Pr\left[ \left\| \mathcal{W} - \mathcal{I} \right\|_{\diamond}\leq\eps \right] &= \Pr\left[ W \in \mathsf{E}_{\eps} \right] = \Pr\left[ \bmt \in \widetilde{\mathsf{E}}_{\eps} \right] \\
&>1 - C \cdot \frac{8d(d + 1)(2d + 1)}{3\pi (L+1)^3 \eps^3} \\
\overset{(\ast)}&{\geq} 1 - \frac{C}{192 \pi^4} \frac{(d+1)(2d+1)}{d^2} \\
&= 1 - \frac{\pi^2}{9(16 - \pi^2)^2} \frac{(d+1)(2d + 1)}{d^2} \geq 1 - \frac{2\pi^2}{3(16 - \pi^2)^2} > \frac{2}{3},
\end{aligned}
$$
where $(\ast)$ follows from $L + 1 \geq \frac{8\pi d}{\eps}$ and $d \geq 1$.
This establishes the required success probability for the $\eps$-closeness of the estimate in diamond distance, and the total queries to $U$ gives $Q = d(L - d)$ [cf. Corollary \ref{cor:decomp_ext_space_with_box_number}]. This completes the proof of \tref{thm:query_optimal_parallel_tomo}.

We close this section by supplying the derivation of equality $(\ast\ast)$ in
\eref{eqn:expanding_by_marginal} as remarked in the previous context. For $S = \{s_2 < \cdots < s_d\}$, let $\mathfrak{S}(S)$ be the set of orderings $\sigma = (\sigma_2, \dots, \sigma_d)$ of $S$ where $\sgn(\sigma) = \sgn(\pi)$ for the permutation $\pi$ of $\{2, \dots, d\}$ such that $\sigma_t = s_{\pi(t)}$. Expanding by definition gives
\begin{align*}
    &\int_{\T^{d-1}} \overline{\det\left[ \varphi_{r}(\theta_t) \right]_{\substack{2 \leq t \leq d \\ r \in [d] \setminus \{\ell\} }}} \det\left[ \varphi_{r'}(\theta_{t'}) \right]_{\substack{2 \leq t' \leq d \\r' \in [d] \setminus \{\ell'\} }} \dd \nu^{\otimes (d-1)}(\bm{\theta}_{\overline{1}}) \\
    &= \int_{\T^{d-1}} \sum_{\substack{ \sigma \in \Sym([d] \setminus \{\ell\} ) \\ \tau \in \Sym([d] \setminus \{\ell'\} )  }} \sgn(\sigma) \sgn(\tau) \prod_{t=2}^{d} \left[ \frac{2}{L+1} \sum_{k, k' = 0}^{L- 1} \sin\left( \frac{(k + 1) \sigma_t}{L+1} \pi \right) \sin \left( \frac{(k' + 1) \tau_t}{L+1} \pi \right) e^{\I (k' - k) \theta_t}  \right] \dd \nu^{\otimes (d-1)} (\bm{\theta}_{\overline{1}}) \\
    &= \sum_{\substack{ \sigma \in \Sym([d] \setminus \{\ell\} ) \\ \tau \in \Sym([d] \setminus \{\ell'\} )  }} \sgn(\sigma) \sgn(\tau) \prod_{t=2}^{d} \left[ \frac{2}{L+1} \sum_{k, k' = 0}^{L- 1} \sin\left( \frac{(k + 1) \sigma_t}{L+1} \pi \right) \sin \left( \frac{(k' + 1) \tau_t}{L+1} \pi \right)  \int_{\T} e^{\I (k' - k) \theta_t} \dd \nu(\theta_t) \right] \\
    &= \sum_{\substack{ \sigma \in \Sym([d] \setminus \{\ell\} ) \\ \tau \in \Sym([d] \setminus \{\ell'\} )  }} \sgn(\sigma) \sgn(\tau) \prod_{t=2}^{d} \left[ \frac{2}{L+1} \sum_{k = 0}^{L- 1} \sin\left( \frac{(k + 1) \sigma_t}{L+1} \pi \right) \sin \left( \frac{(k + 1) \tau_t}{L+1} \pi \right)  \right] \\
    &= \sum_{\substack{ \sigma \in \Sym([d] \setminus \{\ell\} ) \\ \tau \in \Sym([d] \setminus \{\ell'\} )  }} \sgn(\sigma) \sgn(\tau) \prod_{t=2}^d \delta_{\sigma_t, \tau_t} = \sum_{\substack{ \sigma \in \Sym([d] \setminus \{\ell\} ) \\ \tau \in \Sym([d] \setminus \{\ell'\} )  }} \sgn(\sigma) \sgn(\tau) \delta_{\sigma, \tau} = (d - 1)! \cdot \delta_{ [d] \setminus \{\ell\}, [d] \setminus \{\ell'\} }.
\end{align*}

\section*{Acknowledgments}
We are grateful to Zi-Shen Li for insightful discussions that inspired the application of our protocol to the tomography of fermionic linear optics. We also thank Noam Scully and Sisi Zhou for helpful discussions. The amplitude distribution of the probe state was inspired by Holevo's covariant phase estimation algorithm~\cite{Holevo2011}: we reformulated the likelihood function that arises from the covariant learning framework of Bisio et~al.~\cite{PhysRevA.81.032324} and prompted ChatGPT 5.6 with the idea of relating the additive Holevo cost minimization to a ground-state problem restricted to the antisymmetric subspace. We subsequently fitted the construction into Bisio et al.'s canonical framework, bounded the success probability from below, and checked the validity with assistance from ChatGPT. 

This work is supported by the National Natural Science Foundation of China via the Excellent Young Scientists Fund (Hong Kong and Macau) Project 12322516, the National Natural Science Foundation of China (NSFC)/Research Grants Council (RGC) Joint Research Scheme via Project N\_HKU7107/24, the Guangdong Provincial Quantum Science Strategic Initiative via Project GDZX2503001 and the Hong Kong Research Grants Council (RGC) through the General Research Fund (GRF) Grant 17305625.

\bibliographystyle{alphaurl}
\bibliography{opt_par_tomo_updated}

\newcommand{\etalchar}[1]{$^{#1}$}
\begin{thebibliography}{BBMT04b}

\bibitem[AAB{\etalchar{+}}19]{arute2019quantum}
Frank Arute, Kunal Arya, Ryan Babbush, et~al.
\newblock Quantum supremacy using a programmable superconducting processor.
\newblock {\em Nature}, 574(7779):505--510, 2019.
\newblock \href {https://doi.org/10.1038/s41586-019-1666-5} {\path{doi:10.1038/s41586-019-1666-5}}.

\bibitem[AJV01]{PhysRevA.64.050302}
A.~Ac{\'\i}n, E.~Jan{\'e}, and G.~Vidal.
\newblock Optimal estimation of quantum dynamics.
\newblock {\em Physical Review A}, 64(5):050302(R), 2001.
\newblock \href {https://arxiv.org/abs/quant-ph/0012015} {\path{arXiv:quant-ph/0012015}}, \href {https://doi.org/10.1103/PhysRevA.64.050302} {\path{doi:10.1103/PhysRevA.64.050302}}.

\bibitem[BBMT04a]{BaganBaigMunozTapia2004a}
E.~Bagan, M.~Baig, and R.~Mu{\~n}oz-Tapia.
\newblock Entanglement-assisted alignment of reference frames using a dense covariant coding.
\newblock {\em Physical Review A}, 69(5):050303(R), 2004.
\newblock \href {https://arxiv.org/abs/quant-ph/0303019} {\path{arXiv:quant-ph/0303019}}, \href {https://doi.org/10.1103/PhysRevA.69.050303} {\path{doi:10.1103/PhysRevA.69.050303}}.

\bibitem[BBMT04b]{BaganBaigMunozTapia2004b}
E.~Bagan, M.~Baig, and R.~Mu{\~n}oz-Tapia.
\newblock Quantum reverse engineering and reference-frame alignment without nonlocal correlations.
\newblock {\em Physical Review A}, 70(3):030301(R), 2004.
\newblock \href {https://arxiv.org/abs/quant-ph/0405082} {\path{arXiv:quant-ph/0405082}}, \href {https://doi.org/10.1103/PhysRevA.70.030301} {\path{doi:10.1103/PhysRevA.70.030301}}.

\bibitem[BCD{\etalchar{+}}10]{PhysRevA.81.032324}
Alessandro Bisio, Giulio Chiribella, Giacomo~Mauro D'Ariano, Stefano Facchini, and Paolo Perinotti.
\newblock Optimal quantum learning of a unitary transformation.
\newblock {\em Physical Review A}, 81(3):032324, 2010.
\newblock \href {https://arxiv.org/abs/0903.0543} {\path{arXiv:0903.0543}}, \href {https://doi.org/10.1103/PhysRevA.81.032324} {\path{doi:10.1103/PhysRevA.81.032324}}.

\bibitem[BCH07]{schur_transform_bacon}
Dave Bacon, Isaac~L. Chuang, and Aram~W. Harrow.
\newblock The quantum {Schur} and {Clebsch--Gordan} transforms: {I}. efficient qudit circuits.
\newblock In {\em Proceedings of the Eighteenth Annual {ACM--SIAM} Symposium on Discrete Algorithms}, pages 1235--1244, Philadelphia, PA, 2007. Society for Industrial and Applied Mathematics.
\newblock \href {https://arxiv.org/abs/quant-ph/0601001} {\path{arXiv:quant-ph/0601001}}.

\bibitem[BDL{\etalchar{+}}26]{braccia2026commutantfermionicgaussianunitaries}
Paolo Braccia, N.~L. Diaz, Martin Larocca, M.~Cerezo, and Diego Garc{\'i}a-Mart{\'i}n.
\newblock The commutant of fermionic {Gaussian} unitaries, 2026.
\newblock \href {https://arxiv.org/abs/2603.19210} {\path{arXiv:2603.19210}}.

\bibitem[BDM99]{Bu_ek_1999}
Vladim{\'i}r Bu{\v{z}}ek, Radoslav Derka, and Serge Massar.
\newblock Optimal quantum clocks.
\newblock {\em Physical Review Letters}, 82(10):2207--2210, 1999.
\newblock \href {https://arxiv.org/abs/quant-ph/9808042} {\path{arXiv:quant-ph/9808042}}, \href {https://doi.org/10.1103/PhysRevLett.82.2207} {\path{doi:10.1103/PhysRevLett.82.2207}}.

\bibitem[BEG{\etalchar{+}}24]{bluvstein2024logical}
Dolev Bluvstein, Simon~J. Evered, Alexandra~A. Geim, et~al.
\newblock Logical quantum processor based on reconfigurable atom arrays.
\newblock {\em Nature}, 626(7997):58--65, 2024.
\newblock \href {https://arxiv.org/abs/2312.03982} {\path{arXiv:2312.03982}}, \href {https://doi.org/10.1038/s41586-023-06927-3} {\path{doi:10.1038/s41586-023-06927-3}}.

\bibitem[BFG{\etalchar{+}}25]{burchardt2025highdimensionalquantumschurtransforms}
Adam Burchardt, Jiani Fei, Dmitry Grinko, Martin Larocca, Maris Ozols, Sydney Timmerman, and Vladyslav Visnevskyi.
\newblock High-dimensional quantum {Schur} transforms, 2025.
\newblock \href {https://arxiv.org/abs/2509.22640} {\path{arXiv:2509.22640}}.

\bibitem[Bha97]{Bhatia1997}
Rajendra Bhatia.
\newblock {\em Matrix Analysis}, volume 169 of {\em Graduate Texts in Mathematics}.
\newblock Springer, 1997.
\newblock \href {https://doi.org/10.1007/978-1-4612-0653-8} {\path{doi:10.1007/978-1-4612-0653-8}}.

\bibitem[BKD14]{BaldwinKalevDeutsch2014}
Charles~H. Baldwin, Amir Kalev, and Ivan~H. Deutsch.
\newblock Quantum process tomography of unitary and near-unitary maps.
\newblock {\em Physical Review A}, 90(1):012110, 2014.
\newblock \href {https://arxiv.org/abs/1404.2877} {\path{arXiv:1404.2877}}, \href {https://doi.org/10.1103/PhysRevA.90.012110} {\path{doi:10.1103/PhysRevA.90.012110}}.

\bibitem[BMQ21]{PhysRevLett.127.200504}
Jessica Bavaresco, Mio Murao, and Marco~T{\'u}lio Quintino.
\newblock Strict hierarchy between parallel, sequential, and indefinite-causal-order strategies for channel discrimination.
\newblock {\em Physical Review Letters}, 127(20):200504, 2021.
\newblock \href {https://arxiv.org/abs/2011.08300} {\path{arXiv:2011.08300}}, \href {https://doi.org/10.1103/PhysRevLett.127.200504} {\path{doi:10.1103/PhysRevLett.127.200504}}.

\bibitem[BMQ22]{Bavaresco2022}
Jessica Bavaresco, Mio Murao, and Marco~T{\'u}lio Quintino.
\newblock Unitary channel discrimination beyond group structures: Advantages of sequential and indefinite-causal-order strategies.
\newblock {\em Journal of Mathematical Physics}, 63(4):042203, 2022.
\newblock \href {https://arxiv.org/abs/2105.13369} {\path{arXiv:2105.13369}}, \href {https://doi.org/10.1063/5.0075919} {\path{doi:10.1063/5.0075919}}.

\bibitem[Bra05]{Bravyi_Lagrangian_2005}
Sergey Bravyi.
\newblock Lagrangian representation for fermionic linear optics.
\newblock {\em Quantum Information and Computation}, 5(3):216--238, 2005.
\newblock \href {https://arxiv.org/abs/quant-ph/0404180} {\path{arXiv:quant-ph/0404180}}, \href {https://doi.org/10.26421/QIC5.3-3} {\path{doi:10.26421/QIC5.3-3}}.

\bibitem[CDP08]{PhysRevLett.101.180501}
Giulio Chiribella, Giacomo~M. D'Ariano, and Paolo Perinotti.
\newblock Memory effects in quantum channel discrimination.
\newblock {\em Physical Review Letters}, 101(18):180501, 2008.
\newblock \href {https://arxiv.org/abs/0803.3237} {\path{arXiv:0803.3237}}, \href {https://doi.org/10.1103/PhysRevLett.101.180501} {\path{doi:10.1103/PhysRevLett.101.180501}}.

\bibitem[CDP09]{Giulio_comb_2009}
Giulio Chiribella, Giacomo~Mauro D'Ariano, and Paolo Perinotti.
\newblock Theoretical framework for quantum networks.
\newblock {\em Physical Review A}, 80(2):022339, 2009.
\newblock \href {https://arxiv.org/abs/0904.4483} {\path{arXiv:0904.4483}}, \href {https://doi.org/10.1103/PhysRevA.80.022339} {\path{doi:10.1103/PhysRevA.80.022339}}.

\bibitem[CDPS04]{chiribella2004efficient}
G.~Chiribella, G.~M. D'Ariano, P.~Perinotti, and M.~F. Sacchi.
\newblock Efficient use of quantum resources for the transmission of a reference frame.
\newblock {\em Physical Review Letters}, 93(18):180503, 2004.
\newblock \href {https://arxiv.org/abs/quant-ph/0405095} {\path{arXiv:quant-ph/0405095}}, \href {https://doi.org/10.1103/PhysRevLett.93.180503} {\path{doi:10.1103/PhysRevLett.93.180503}}.

\bibitem[CDS05]{PhysRevA.72.042338}
G.~Chiribella, G.~M. D'Ariano, and M.~F. Sacchi.
\newblock Optimal estimation of group transformations using entanglement.
\newblock {\em Physical Review A}, 72(4):042338, 2005.
\newblock \href {https://arxiv.org/abs/quant-ph/0506267} {\path{arXiv:quant-ph/0506267}}, \href {https://doi.org/10.1103/PhysRevA.72.042338} {\path{doi:10.1103/PhysRevA.72.042338}}.

\bibitem[CGO{\etalchar{+}}26]{chen_Girardi2026}
Kean Chen, Filippo Girardi, Aadil Oufkir, Nengkun Yu, and Zhicheng Zhang.
\newblock Quantum channel tomography: Optimal bounds and a {Heisenberg}-to-classical phase transition, 2026.
\newblock \href {https://arxiv.org/abs/2604.17369} {\path{arXiv:2604.17369}}.

\bibitem[CHL{\etalchar{+}}23]{10353129}
Sitan Chen, Brice Huang, Jerry Li, Allen Liu, and Mark Sellke.
\newblock When does adaptivity help for quantum state learning?
\newblock In {\em 2023 IEEE 64th Annual Symposium on Foundations of Computer Science ({FOCS})}, pages 391--404. IEEE, 2023.
\newblock \href {https://arxiv.org/abs/2206.05265} {\path{arXiv:2206.05265}}, \href {https://doi.org/10.1109/FOCS57990.2023.00029} {\path{doi:10.1109/FOCS57990.2023.00029}}.

\bibitem[CLM{\etalchar{+}}21]{Christandl2021}
Matthias Christandl, Felix Leditzky, Christian Majenz, Graeme Smith, Florian Speelman, and Michael Walter.
\newblock Asymptotic performance of port-based teleportation.
\newblock {\em Communications in Mathematical Physics}, 381(1):379--451, 2021.
\newblock \href {https://arxiv.org/abs/1809.10751} {\path{arXiv:1809.10751}}, \href {https://doi.org/10.1007/s00220-020-03884-0} {\path{doi:10.1007/s00220-020-03884-0}}.

\bibitem[CS24]{Cudby2024}
Josh Cudby and Sergii Strelchuk.
\newblock Learning {Gaussian} operations and the {Matchgate} hierarchy.
\newblock In {\em 2024 IEEE International Conference on Quantum Computing and Engineering ({QCE})}, pages 141--149. IEEE, 2024.
\newblock \href {https://arxiv.org/abs/2407.12649} {\path{arXiv:2407.12649}}, \href {https://doi.org/10.1109/QCE60285.2024.00026} {\path{doi:10.1109/QCE60285.2024.00026}}.

\bibitem[CZ26]{christensen2026learningfermioniclinearoptics}
Aria Christensen and Andrew Zhao.
\newblock Learning fermionic linear optics with {Heisenberg} scaling and physical operations, 2026.
\newblock \href {https://arxiv.org/abs/2602.05058} {\path{arXiv:2602.05058}}.

\bibitem[DFY07]{PhysRevLett.98.100503}
Runyao Duan, Yuan Feng, and Mingsheng Ying.
\newblock Entanglement is not necessary for perfect discrimination between unitary operations.
\newblock {\em Physical Review Letters}, 98(10):100503, 2007.
\newblock \href {https://arxiv.org/abs/quant-ph/0601150} {\path{arXiv:quant-ph/0601150}}, \href {https://doi.org/10.1103/PhysRevLett.98.100503} {\path{doi:10.1103/PhysRevLett.98.100503}}.

\bibitem[DG13]{Dereziński_Gérard_2013}
Jan Derezi{\'n}ski and Christian G{\'e}rard.
\newblock {\em Mathematics of Quantization and Quantum Fields}.
\newblock Cambridge Monographs on Mathematical Physics. Cambridge University Press, Cambridge, 2013.
\newblock \href {https://doi.org/10.1017/CBO9780511894541} {\path{doi:10.1017/CBO9780511894541}}.

\bibitem[GL25]{grewal2026efficientlearningstructuredquantum}
Sabee Grewal and Daniel Liang.
\newblock Efficient learning of structured quantum circuits via {Pauli} dimensionality and sparsity, 2025.
\newblock \href {https://arxiv.org/abs/2510.00168} {\path{arXiv:2510.00168}}.

\bibitem[GLM06]{giovannetti2006quantum}
Vittorio Giovannetti, Seth Lloyd, and Lorenzo Maccone.
\newblock Quantum metrology.
\newblock {\em Physical Review Letters}, 96(1):010401, 2006.
\newblock \href {https://arxiv.org/abs/quant-ph/0509179} {\path{arXiv:quant-ph/0509179}}, \href {https://doi.org/10.1103/PhysRevLett.96.010401} {\path{doi:10.1103/PhysRevLett.96.010401}}.

\bibitem[GLM11]{giovannetti2011advances}
Vittorio Giovannetti, Seth Lloyd, and Lorenzo Maccone.
\newblock Advances in quantum metrology.
\newblock {\em Nature Photonics}, 5(4):222--229, 2011.
\newblock \href {https://arxiv.org/abs/1102.2318} {\path{arXiv:1102.2318}}, \href {https://doi.org/10.1038/nphoton.2011.35} {\path{doi:10.1038/nphoton.2011.35}}.

\bibitem[GMZ{\etalchar{+}}25]{girardi2025random}
Filippo Girardi, Francesco~Anna Mele, Haimeng Zhao, Marco Fanizza, and Ludovico Lami.
\newblock Random {Stinespring} superchannel: Converting channel queries into dilation isometry queries, 2025.
\newblock \href {https://arxiv.org/abs/2512.20599} {\path{arXiv:2512.20599}}.

\bibitem[Gri25]{grinko2025mixed}
Dmitry~A. Grinko.
\newblock {\em Mixed {Schur--Weyl} duality in quantum information}.
\newblock PhD thesis, Universiteit van Amsterdam, Amsterdam, 2025.
\newblock ILLC Dissertation Series DS-2025-02.
\newblock URL: \url{https://hdl.handle.net/11245.1/d9e16c26-ef20-40c5-b847-53c13d1a8a1a}.

\bibitem[GSTW24]{10.1145/3618260.3649621}
Uma Girish, Makrand Sinha, Avishay Tal, and Kewen Wu.
\newblock The power of adaptivity in quantum query algorithms.
\newblock In {\em Proceedings of the 56th Annual {ACM} Symposium on Theory of Computing}, pages 1488--1497, New York, NY, USA, 2024. Association for Computing Machinery.
\newblock \href {https://arxiv.org/abs/2311.16057} {\path{arXiv:2311.16057}}, \href {https://doi.org/10.1145/3618260.3649621} {\path{doi:10.1145/3618260.3649621}}.

\bibitem[GW09]{Goodman2009}
Roe Goodman and Nolan~R. Wallach.
\newblock {\em Symmetry, Representations, and Invariants}, volume 255 of {\em Graduate Texts in Mathematics}.
\newblock Springer, 2009.
\newblock \href {https://doi.org/10.1007/978-0-387-79852-3} {\path{doi:10.1007/978-0-387-79852-3}}.

\bibitem[Har05]{Har05}
Aram~W. Harrow.
\newblock {\em Applications of coherent classical communication and the {Schur} transform to quantum information theory}.
\newblock PhD thesis, Massachusetts Institute of Technology, Cambridge, MA, 2005.
\newblock URL: \url{https://dspace.mit.edu/entities/publication/a43d9288-c946-4317-b408-b33136c08977}, \href {https://arxiv.org/abs/quant-ph/0512255} {\path{arXiv:quant-ph/0512255}}.

\bibitem[Hay06]{hayashi2006parallel}
Masahito Hayashi.
\newblock Parallel treatment of estimation of {SU(2)} and phase estimation.
\newblock {\em Physics Letters A}, 354(3):183--189, 2006.
\newblock \href {https://arxiv.org/abs/quant-ph/0407053} {\path{arXiv:quant-ph/0407053}}, \href {https://doi.org/10.1016/j.physleta.2006.01.043} {\path{doi:10.1016/j.physleta.2006.01.043}}.

\bibitem[Hay25]{Hayashi2025indefinitecausal}
Masahito Hayashi.
\newblock Indefinite causal order strategy does not improve the estimation of group action.
\newblock {\em Quantum}, 9:1891, 2025.
\newblock \href {https://arxiv.org/abs/2501.09312} {\path{arXiv:2501.09312}}, \href {https://doi.org/10.22331/q-2025-10-22-1891} {\path{doi:10.22331/q-2025-10-22-1891}}.

\bibitem[Hel69]{Helstrom1969}
Carl~W. Helstrom.
\newblock Quantum detection and estimation theory.
\newblock {\em Journal of Statistical Physics}, 1(2):231--252, 1969.
\newblock \href {https://doi.org/10.1007/BF01007479} {\path{doi:10.1007/BF01007479}}.

\bibitem[HHLW10]{PhysRevA.81.032339}
Aram~W. Harrow, Avinatan Hassidim, Debbie~W. Leung, and John Watrous.
\newblock Adaptive versus nonadaptive strategies for quantum channel discrimination.
\newblock {\em Physical Review A}, 81(3):032339, 2010.
\newblock \href {https://arxiv.org/abs/0909.0256} {\path{arXiv:0909.0256}}, \href {https://doi.org/10.1103/PhysRevA.81.032339} {\path{doi:10.1103/PhysRevA.81.032339}}.

\bibitem[HJ12]{Horn2012}
Roger~A. Horn and Charles~R. Johnson.
\newblock {\em Matrix Analysis}.
\newblock Cambridge University Press, second edition, 2012.
\newblock \href {https://doi.org/10.1017/CBO9781139020411} {\path{doi:10.1017/CBO9781139020411}}.

\bibitem[HKOT23]{HKOT23}
Jeongwan Haah, Robin Kothari, Ryan O'Donnell, and Ewin Tang.
\newblock Query-optimal estimation of unitary channels in diamond distance.
\newblock In {\em 2023 {IEEE} 64th Annual Symposium on Foundations of Computer Science ({FOCS})}, pages 363--390. IEEE, 2023.
\newblock \href {https://arxiv.org/abs/2302.14066} {\path{arXiv:2302.14066}}, \href {https://doi.org/10.1109/FOCS57990.2023.00028} {\path{doi:10.1109/FOCS57990.2023.00028}}.

\bibitem[HLS{\etalchar{+}}26]{he2026optimalclassicalshadowestimation}
Entong He, Zihao Li, Noam Scully, Sisi Zhou, and Yuxiang Yang.
\newblock Optimal classical shadow estimation of unitary channels at the {Heisenberg} limit, 2026.
\newblock \href {https://arxiv.org/abs/2606.13638} {\path{arXiv:2606.13638}}.

\bibitem[HNG{\etalchar{+}}25]{PRXQuantum.6.030202}
Akel Hashim, Long~B. Nguyen, Noah Goss, Brian Marinelli, Ravi~K. Naik, Trevor Chistolini, Jordan Hines, J.~P. Marceaux, Yosep Kim, Pranav Gokhale, Teague Tomesh, Senrui Chen, Liang Jiang, Samuele Ferracin, Kenneth Rudinger, Timothy Proctor, Kevin~C. Young, Irfan Siddiqi, and Robin Blume-Kohout.
\newblock Practical introduction to benchmarking and characterization of quantum computers.
\newblock {\em PRX Quantum}, 6(3):030202, 2025.
\newblock \href {https://arxiv.org/abs/2408.12064} {\path{arXiv:2408.12064}}, \href {https://doi.org/10.1103/PRXQuantum.6.030202} {\path{doi:10.1103/PRXQuantum.6.030202}}.

\bibitem[Hol11]{Holevo2011}
Alexander~S. Holevo.
\newblock {\em Probabilistic and Statistical Aspects of Quantum Theory}, volume~1 of {\em Publications of the Scuola Normale Superiore}.
\newblock Edizioni della Normale, Pisa, first edition, 2011.
\newblock \href {https://doi.org/10.1007/978-88-7642-378-9} {\path{doi:10.1007/978-88-7642-378-9}}.

\bibitem[How87]{Howe1987}
Roger Howe.
\newblock ({GL}n, {GL}m)-duality and symmetric plethysm.
\newblock {\em Proceedings of the Indian Academy of Sciences: Mathematical Sciences}, 97(1--3):85--109, 1987.
\newblock \href {https://doi.org/10.1007/BF02837817} {\path{doi:10.1007/BF02837817}}.

\bibitem[HP19]{hennessy2019golden}
John~L. Hennessy and David~A. Patterson.
\newblock A new golden age for computer architecture.
\newblock {\em Communications of the ACM}, 62(2):48--60, 2019.
\newblock \href {https://doi.org/10.1145/3282307} {\path{doi:10.1145/3282307}}.

\bibitem[HY26]{He2026resource}
Entong He and Yuxiang Yang.
\newblock Resource quantification for programming low-depth quantum circuits.
\newblock {\em Quantum}, 10:2166, 2026.
\newblock \href {https://arxiv.org/abs/2509.09642} {\path{arXiv:2509.09642}}, \href {https://doi.org/10.22331/q-2026-07-20-2166} {\path{doi:10.22331/q-2026-07-20-2166}}.

\bibitem[Kah07]{Kahn_2007}
Jonas Kahn.
\newblock Fast rate estimation of a unitary operation in {$\mathrm{SU}(d)$}.
\newblock {\em Physical Review A}, 75(2):022326, 2007.
\newblock \href {https://arxiv.org/abs/quant-ph/0603115} {\path{arXiv:quant-ph/0603115}}, \href {https://doi.org/10.1103/PhysRevA.75.022326} {\path{doi:10.1103/PhysRevA.75.022326}}.

\bibitem[Kak37]{kakutani1937beweis}
Shizuo Kakutani.
\newblock Ein beweis des satzes von {M. Eidelheit} {\"u}ber konvexe mengen.
\newblock {\em Proceedings of the Imperial Academy}, 13(4):93--94, 1937.
\newblock \href {https://doi.org/10.3792/pia/1195579980} {\path{doi:10.3792/pia/1195579980}}.

\bibitem[Kro19]{Krovi2019efficienthigh}
Hari Krovi.
\newblock An efficient high dimensional quantum {Schur} transform.
\newblock {\em Quantum}, 3:122, 2019.
\newblock \href {https://arxiv.org/abs/1804.00055} {\path{arXiv:1804.00055}}, \href {https://doi.org/10.22331/q-2019-02-14-122} {\path{doi:10.22331/q-2019-02-14-122}}.

\bibitem[LBH15]{lecun2015deep}
Yann LeCun, Yoshua Bengio, and Geoffrey Hinton.
\newblock Deep learning.
\newblock {\em Nature}, 521(7553):436--444, 2015.
\newblock \href {https://doi.org/10.1038/nature14539} {\path{doi:10.1038/nature14539}}.

\bibitem[Leu00]{leung2000robustquantumcomputation}
Debbie~W. Leung.
\newblock {\em Towards Robust Quantum Computation}.
\newblock PhD thesis, Stanford University, Stanford, CA, 2000.
\newblock \href {https://arxiv.org/abs/cs/0012017} {\path{arXiv:cs/0012017}}.

\bibitem[LHYY23]{Liu_2023}
Qiushi Liu, Zihao Hu, Haidong Yuan, and Yuxiang Yang.
\newblock Optimal strategies of quantum metrology with a strict hierarchy.
\newblock {\em Physical Review Letters}, 130(7):070803, 2023.
\newblock \href {https://arxiv.org/abs/2203.09758} {\path{arXiv:2203.09758}}, \href {https://doi.org/10.1103/PhysRevLett.130.070803} {\path{doi:10.1103/PhysRevLett.130.070803}}.

\bibitem[MB25]{mele2026optimallearningquantumchannels}
Antonio~Anna Mele and Lennart Bittel.
\newblock Optimal learning of quantum channels in diamond distance, 2025.
\newblock \href {https://arxiv.org/abs/2512.10214} {\path{arXiv:2512.10214}}.

\bibitem[NC10]{Nielsen_Chuang_2010}
Michael~A. Nielsen and Isaac~L. Chuang.
\newblock {\em Quantum Computation and Quantum Information}.
\newblock Cambridge University Press, 10th anniversary edition, 2010.
\newblock \href {https://doi.org/10.1017/CBO9780511976667} {\path{doi:10.1017/CBO9780511976667}}.

\bibitem[ODMZ22]{PRXQuantum.3.020328}
Micha{\l} Oszmaniec, Ninnat Dangniam, Mauro E.~S. Morales, and Zolt{\'a}n Zimbor{\'a}s.
\newblock Fermion sampling: A robust quantum computational advantage scheme using fermionic linear optics and magic input states.
\newblock {\em PRX Quantum}, 3(2):020328, 2022.
\newblock \href {https://arxiv.org/abs/2012.15825} {\path{arXiv:2012.15825}}, \href {https://doi.org/10.1103/PRXQuantum.3.020328} {\path{doi:10.1103/PRXQuantum.3.020328}}.

\bibitem[Oko01]{Okounkov2001}
Andrei Okounkov.
\newblock Infinite wedge and random partitions.
\newblock {\em Selecta Mathematica, New Series}, 7(1):57--81, 2001.
\newblock \href {https://arxiv.org/abs/math/9907127} {\path{arXiv:math/9907127}}, \href {https://doi.org/10.1007/PL00001398} {\path{doi:10.1007/PL00001398}}.

\bibitem[Pre18]{preskill2018nisq}
John Preskill.
\newblock Quantum computing in the {NISQ} era and beyond.
\newblock {\em Quantum}, 2:79, 2018.
\newblock \href {https://arxiv.org/abs/1801.00862} {\path{arXiv:1801.00862}}, \href {https://doi.org/10.22331/q-2018-08-06-79} {\path{doi:10.22331/q-2018-08-06-79}}.

\bibitem[Pro78]{proskuryakov1974problems}
Igor~V. Proskuryakov.
\newblock {\em Problems in Linear Algebra}.
\newblock Mir Publishers, Moscow, 1978.
\newblock Translated from the Russian by George Yankovsky.

\bibitem[PS02]{peres2001entangled}
Asher Peres and Petra~F. Scudo.
\newblock Covariant quantum measurements may not be optimal.
\newblock {\em Journal of Modern Optics}, 49(8):1235--1243, 2002.
\newblock \href {https://arxiv.org/abs/quant-ph/0107114} {\path{arXiv:quant-ph/0107114}}, \href {https://doi.org/10.1080/09500340110118449} {\path{doi:10.1080/09500340110118449}}.

\bibitem[SBJ19]{Subramanian_2019}
Sathyawageeswar Subramanian, Stephen Brierley, and Richard Jozsa.
\newblock Implementing smooth functions of a {Hermitian} matrix on a quantum computer.
\newblock {\em Journal of Physics Communications}, 3(6):065002, 2019.
\newblock \href {https://arxiv.org/abs/1806.06885} {\path{arXiv:1806.06885}}, \href {https://doi.org/10.1088/2399-6528/ab25a2} {\path{doi:10.1088/2399-6528/ab25a2}}.

\bibitem[SBZ19]{Sedl_k_2019}
Michal Sedl{\'a}k, Alessandro Bisio, and M{\'a}rio Ziman.
\newblock Optimal probabilistic storage and retrieval of unitary channels.
\newblock {\em Physical Review Letters}, 122(17):170502, 2019.
\newblock \href {https://arxiv.org/abs/1809.04552} {\path{arXiv:1809.04552}}, \href {https://doi.org/10.1103/PhysRevLett.122.170502} {\path{doi:10.1103/PhysRevLett.122.170502}}.

\bibitem[Str99]{doi:10.1137/S0036144598336745}
Gilbert Strang.
\newblock The discrete cosine transform.
\newblock {\em SIAM Review}, 41(1):135--147, 1999.
\newblock \href {https://doi.org/10.1137/S0036144598336745} {\path{doi:10.1137/S0036144598336745}}.

\bibitem[SvL13]{stefanucci2013nonequilibrium}
Gianluca Stefanucci and Robert van Leeuwen.
\newblock {\em Nonequilibrium Many-Body Theory of Quantum Systems: A Modern Introduction}.
\newblock Cambridge University Press, Cambridge, 2013.
\newblock \href {https://doi.org/10.1017/CBO9781139023979} {\path{doi:10.1017/CBO9781139023979}}.

\bibitem[SZ26]{NoamZhouNearOptimal2026}
Noam Scully and Sisi Zhou.
\newblock to appear, 2026.

\bibitem[vACGN23]{vanapeldoorn2022quantumtomographyusingstatepreparation}
Joran van Apeldoorn, Arjan Cornelissen, Andr{\'a}s Gily{\'e}n, and Giacomo Nannicini.
\newblock Quantum tomography using state-preparation unitaries.
\newblock In {\em Proceedings of the 2023 Annual ACM--SIAM Symposium on Discrete Algorithms ({SODA})}, pages 1265--1318. Society for Industrial and Applied Mathematics, 2023.
\newblock \href {https://arxiv.org/abs/2207.08800} {\path{arXiv:2207.08800}}, \href {https://doi.org/10.1137/1.9781611977554.ch47} {\path{doi:10.1137/1.9781611977554.ch47}}.

\bibitem[vDDE{\etalchar{+}}07]{vanDam_2007}
Wim van Dam, G.~Mauro D'Ariano, Artur Ekert, Chiara Macchiavello, and Michele Mosca.
\newblock Optimal phase estimation in quantum networks.
\newblock {\em Journal of Physics A: Mathematical and Theoretical}, 40(28):7971--7984, 2007.
\newblock \href {https://arxiv.org/abs/0706.4412} {\path{arXiv:0706.4412}}, \href {https://doi.org/10.1088/1751-8113/40/28/S07} {\path{doi:10.1088/1751-8113/40/28/S07}}.

\bibitem[Wat18]{Watrous2018}
John Watrous.
\newblock {\em The Theory of Quantum Information}.
\newblock Cambridge University Press, 2018.
\newblock \href {https://doi.org/10.1017/9781316848142} {\path{doi:10.1017/9781316848142}}.

\bibitem[WBC{\etalchar{+}}21]{wu2021strong}
Yulin Wu, Wan-Su Bao, Sirui Cao, et~al.
\newblock Strong quantum computational advantage using a superconducting quantum processor.
\newblock {\em Physical Review Letters}, 127(18):180501, 2021.
\newblock \href {https://arxiv.org/abs/2106.14734} {\path{arXiv:2106.14734}}, \href {https://doi.org/10.1103/PhysRevLett.127.180501} {\path{doi:10.1103/PhysRevLett.127.180501}}.

\bibitem[Wri16]{Wri16}
John Wright.
\newblock {\em How to learn a quantum state}.
\newblock PhD thesis, Carnegie Mellon University, Pittsburgh, PA, 2016.
\newblock Technical Report CMU-CS-16-108.
\newblock URL: \url{https://csd.cs.cmu.edu/academics/doctoral/degrees-conferred/john-wright}.

\bibitem[Yan19]{PhysRevLett.123.110501}
Yuxiang Yang.
\newblock Memory effects in quantum metrology.
\newblock {\em Physical Review Letters}, 123(11):110501, 2019.
\newblock \href {https://arxiv.org/abs/1904.07267} {\path{arXiv:1904.07267}}, \href {https://doi.org/10.1103/PhysRevLett.123.110501} {\path{doi:10.1103/PhysRevLett.123.110501}}.

\bibitem[YKS{\etalchar{+}}26]{one_to_one_correspondence2026}
Satoshi Yoshida, Yuki Koizumi, Micha{\l} Studzi{\'n}ski, Marco~T{\'u}lio Quintino, and Mio Murao.
\newblock One-to-one correspondence between deterministic port-based teleportation and unitary estimation.
\newblock {\em IEEE Transactions on Information Theory}, 72(4):2358--2377, 2026.
\newblock \href {https://arxiv.org/abs/2408.11902} {\path{arXiv:2408.11902}}, \href {https://doi.org/10.1109/TIT.2026.3658543} {\path{doi:10.1109/TIT.2026.3658543}}.

\bibitem[YMR{\etalchar{+}}22]{Yang_2022}
Yuxiang Yang, Yin Mo, Joseph~M. Renes, Giulio Chiribella, and Mischa~P. Woods.
\newblock Optimal universal quantum error correction via bounded reference frames.
\newblock {\em Physical Review Research}, 4(2):023107, 2022.
\newblock \href {https://arxiv.org/abs/2007.09154} {\path{arXiv:2007.09154}}, \href {https://doi.org/10.1103/PhysRevResearch.4.023107} {\path{doi:10.1103/PhysRevResearch.4.023107}}.

\bibitem[YRC20]{Yang_2020}
Yuxiang Yang, Renato Renner, and Giulio Chiribella.
\newblock Optimal universal programming of unitary gates.
\newblock {\em Physical Review Letters}, 125(21):210501, 2020.
\newblock \href {https://arxiv.org/abs/2007.10363} {\path{arXiv:2007.10363}}, \href {https://doi.org/10.1103/PhysRevLett.125.210501} {\path{doi:10.1103/PhysRevLett.125.210501}}.

\bibitem[YYM25]{yoshida2025asymptoticallyoptimalunitaryestimation}
Satoshi Yoshida, Hironobu Yoshida, and Mio Murao.
\newblock Asymptotically optimal unitary estimation in {$\mathrm{SU}(3)$} by the analysis of graph {Laplacian}, 2025.
\newblock \href {https://arxiv.org/abs/2509.20608} {\path{arXiv:2509.20608}}.

\bibitem[ZLK{\etalchar{+}}24]{ZhaoEtAl2024BoundedGate}
Haimeng Zhao, Laura Lewis, Ishaan Kannan, Yihui Quek, Hsin-Yuan Huang, and Matthias~C. Caro.
\newblock Learning quantum states and unitaries of bounded gate complexity.
\newblock {\em PRX Quantum}, 5(4):040306, 2024.
\newblock \href {https://arxiv.org/abs/2310.19882} {\path{arXiv:2310.19882}}, \href {https://doi.org/10.1103/PRXQuantum.5.040306} {\path{doi:10.1103/PRXQuantum.5.040306}}.

\bibitem[ZWD{\etalchar{+}}20]{zhong2020quantum}
Han-Sen Zhong, Hui Wang, Yu-Hao Deng, et~al.
\newblock Quantum computational advantage using photons.
\newblock {\em Science}, 370(6523):1460--1463, 2020.
\newblock \href {https://arxiv.org/abs/2012.01625} {\path{arXiv:2012.01625}}, \href {https://doi.org/10.1126/science.abe8770} {\path{doi:10.1126/science.abe8770}}.

\end{thebibliography}

\end{document}